\documentclass[AEJ]{AEA}

\usepackage[authoryear]{natbib}

\usepackage{graphicx,subcaption}
\usepackage{multirow}
\usepackage{amsmath,amsfonts,amssymb}
\usepackage{comment}
\newcommand{\matr}[1]{\mathbf{#1}}
\renewcommand{\vec}[1]{\mathbf{#1}}

\DeclareMathOperator{\p}{\mathbb{P}}
\DeclareMathOperator*{\plim}{plim}
\newcommand{\E}{{\rm I}\kern-0.18em{\rm E}}
\newcommand{\N}{{\rm I}\kern-0.18em{\rm N}}

\newtheorem{theorem}{Theorem}
\newtheorem{example}{Example}
\newtheorem{lemma}{Lemma}

\newtheorem{definition}{Definition}
\newtheorem{proposition}{Proposition}
\newtheorem{assumption}{Assumption}
\newtheorem{property}{Property}
\newtheorem{remark}{Remark}[section]

\newcommand{\qed}{\hfill \ensuremath{\Box}}

\draftSpacing{1.5}

\begin{document}

\title{Degree-Weighted Social Learning}

\author{Chen Cheng, Xiao Han, Xin Tong, Yusheng Wu, and Yiqing Xing%
\thanks{
Cheng: Johns Hopkins University, chencheng.ccer@gmail.com; Han: Department of Statistics and Finance, School of Management, University of Science and Technology of China, xhan011@ustc.edu.cn; Tong: Faculty of Business and Economics, the University of Hong Kong, xint@hku.hk, on leave from the Marshall School of Business, University of Southern California; Wu: University of Southern California, yushengmath@gmail.com; Xing: China Center for Economic Research, NSD, Peking University, xingyq@gmail.com. 
An earlier version of this paper was titled ``From Authority-Respect to Grassroots-Dissent: Degree-Weighted Learning and Convergence Speed.'' 
We thank 
Alexander Wolitzky, three anonymous referees, Ben Golub, Matt Jackson, Steven Heilman, David Miller, and participants at numerous conferences and workshops 
for their invaluable suggestions and expert advice, and Qiming Sun for careful proofreading of the paper.  Funding acknowledgment: Discovery Awards from JHU (Cheng), NSFC No.12371278 (Xiao), STARS funding from USC and seed funding from HKU (Tong), and seed funding from PKU (Xing).
All remaining errors are our own.}
}

\date{December 25, 2025}

 \pubMonth{Month}
 \pubYear{Year}
 \pubVolume{Vol}
 \pubIssue{Issue}

\begin{abstract}

We study social learning in which agents weight neighbors' opinions differently based on their degrees, capturing situations in which agents place more trust in well-connected individuals or, conversely, discount their influence. We derive asymptotic properties of learning outcomes in large stochastic networks and analyze how the weighting rule affects societal wisdom and convergence speed. We find that assigning greater weight to higher-degree neighbors harms wisdom but has a non-monotonic effect on convergence speed, depending on the diversity of views within high- and low-degree groups, highlighting a potential trade-off between convergence speed and wisdom. (JEL D83, D85, Z12)

\end{abstract}

\maketitle

\section{Introduction}
\label{sec:intro}

The DeGroot model of social learning in networks has been instrumental in studying societal belief dynamics, such as whether society reaches consensus \citep{degroot1974reaching}, agents' network position and their long-run influence on consensus \citep{demarzo2003persuasion}, the ``wisdom of the crowd'' phenomenon \citep{acemoglu2010spread, GolubJackson2010}, and the speed of belief convergence \citep{GolubJackson2012}. In a typical DeGroot learning model, agents repeatedly aggregate the opinions of their neighbors according to fixed weights over time. To gain tractability—especially when examining learning dynamics in large random networks—agents are often assumed to assign equal weight to the opinions of all their neighbors \citep{GolubJackson2012}. However, the equal-weighting DeGroot model, while a useful benchmark, lacks certain empirically relevant elements. For example, it cannot capture attitudes toward popularity, where society may overweight the opinions of highly connected agents, believing that these individuals are exposed to a larger volume of information. Conversely, equal weighting fails to capture sentiment resistant to popularity, where society may underweight the opinions of highly connected agents, perceiving them as having more redundant or less trustworthy information sources.\footnote{In practice, platforms may intentionally downweight highly popular users to provide ``less-prominent'' users with more visibility. For example, Facebook's News Feed algorithm prioritizes content from friends, family, and close connections over posts from popular pages or celebrities. Similarly, Amazon’s initiatives, such as Amazon Launchpad and A+ Content, support new sellers and startups by offering enhanced product descriptions and media, thereby increasing their exposure and credibility. This approach improves the chances for smaller sellers to compete against established brands.}\footnote{The association between degree and low weighting does not need to be causal for the theory to hold. A high degree alone does not necessarily justify assigning lower weight to someone's influence. Instead, it may be that individuals with certain characteristics (e.g., public figures with broad but less specialized appeal) tend to accumulate a high number of in-links, leading others to assign lower weight to their opinions. We thank the anonymous referee for suggesting this point.}

This paper develops a degree-weighted DeGroot model in which agents assign weights to their neighbors' opinions based on their degree. We allow the weight to increase with the degree of a neighbor, capturing the notion of \emph{popularity-favoring}; conversely, a decreasing weight reflects the sentiment of \emph{popularity-resistance}.\footnote{In our setting, when the weight is independent of degree, the model reduces to the classical equal-weighting case.} Our main contribution is an asymptotic theorem showing that, in our degree-weighted random learning framework, the realized learning outcomes—the consensus belief and the convergence speed—converge to those generated by the ``expected'' learning matrix.\footnote{Previous works, among which \cite{GolubJackson2012}, \cite{dasaratha2020distributions} and \cite{mostagir2021centrality} are the closest, do not address the convergence of highly nonlinear terms arising from degree-dependent weighting.}
Building on this result, we extend the standard equal-weighting DeGroot learning framework to degree-weighted scenarios, introducing a new dimension of comparative statics and an additional degree of freedom in degree dependence, which was previously difficult to explore.

\paragraph{Wisdom} By fully characterizing agents' influence—i.e., how their initial opinions shape the consensus belief—and the consensus belief itself, we apply the necessary and sufficient condition for a society to be wise, defined as the consensus belief uncovering the true state, as established in \cite{GolubJackson2010}. We find that degree-dependent weighting can shift a society between wisdom and its absence.

More specifically, we show that when agents' initial beliefs equal the true state plus i.i.d.\ noise, degree dependence has a \emph{monotonic} effect on societal wisdom. A high degree of popularity-favoring—meaning agents assign greater weight to high-degree neighbors—grants high-degree agents non-vanishing influence and renders the society unwise.


\paragraph{Convergence Speed} 
The degree dependence can have a highly non-monotonic impact on the speed at which societal beliefs reach consensus. We find that when agents exhibit a greater degree of popularity-favoring, society does not necessarily reach consensus faster. Whether increasing the degree of popularity-favoring accelerates or slows convergence depends on the diversity of views within high-degree and low-degree groups. 
When the number of low-degree groups exceeds that of high-degree groups, and different groups may hold different views, increasing the degree of popularity-favoring may accelerate convergence; the reverse occurs when high-degree groups outnumber low-degree groups. 

\paragraph{Trade-off Between Wisdom and Convergence Speed} These forces can create a trade-off between faster convergence and a ``wiser'' society: while a more popularity-resistant society may foster greater wisdom, reaching it can take longer when many low-degree groups exist, as each group potentially holds different views. Similarly, a more popularity-favoring learning rule may accelerate belief convergence at the expense of wisdom. The following table illustrates such a trade-off.\footnote{We discuss the theoretical results underlying this table in more detail in Section~\ref{sec:tradeoff}.}





\begin{table}[ht]
\centering
\caption{Trade-off between Wisdom and Convergence Speed.}
\renewcommand{\arraystretch}{1.3}

\begin{tabular}{ccccc}
\hline\hline

\multicolumn{2}{c}{{\# Groups}} && \multicolumn{2}{c}{Degree Dependence}  \\ \cline{1-2} \cline{4-5}
\multirow{2}*{Low}  & \multirow{2}*{High}  && {Popularity-resistant} & {Popularity-favoring} \\
&  & &  ($\alpha \rightarrow -\infty$) & ($\alpha \rightarrow \infty$)
\\ \hline
1 & 1 && wise, fast & unwise, fast  \\ 
multiple & 1 && wise, slow  &  unwise, fast \\ 
1 & multiple && wise, fast  & unwise, slow  \\ \hline
\end{tabular}

\label{tab:wisdom_speed}
\begin{minipage}{0.95\textwidth}
\smallskip
\footnotesize \textit{{Notes:}} {The first two columns report the number of low- and high-degree groups. 
The parameter $\alpha$ governs the weight that agents place on high-degree neighbors: $\alpha \rightarrow -\infty$ corresponds to an extremely popularity-resistant learning rule, while $\alpha \rightarrow \infty$ corresponds to an extremely popularity-favoring rule. 
The entries summarize whether society converges to wisdom and whether convergence is fast or slow under the corresponding regime.}
\end{minipage}
\end{table}

\subsection*{Literature Review}
\label{sec:lit}

Our paper relates to a substantial body of literature on DeGroot learning in social networks. \cite{degroot1974reaching} was among the first to examine naïve learning, where agents repeatedly aggregate their neighbors' opinions according to a fixed set of weights, and identified conditions for belief convergence to consensus.\footnote{Early contributions include \cite{french1956formal} and \cite{harary1959measurement}. For a recent survey on social learning and a discussion of the distinction between DeGroot and Bayesian models, see \cite{golubs2016}. Subsequent work has extended the DeGroot paradigm to incorporate time-varying weights \cite{chatterjee1977towards}, initial biases or prejudices over agents' prior beliefs \cite{friedkin1990social}, and opinion aggregation restricted to neighbors whose beliefs lie within a certain distance from one's own \cite{Hegselmann2002}, among other variations.} \cite{demarzo2003persuasion} later revisited the DeGroot model in an economic context, grounding naïve learning in a quasi-Bayesian framework and highlighting the role of persuasion bias in belief aggregation, along with other important insights.  
\cite{gao2021centrality} examined centrality-weighted opinion dynamics but did not analyze the impact of such degree dependence on wisdom and convergence speed.  While these studies assumed deterministic networks, we analyze the DeGroot model in large stochastic networks. 

\cite{acemoglu2010spread} and \cite{GolubJackson2010} studied belief convergence to the true state in stochastic networks. \cite{acemoglu2010spread} analyzed how ``forceful agents'' and network structure influence the gap between the consensus belief and the true state, showing that this gap shrinks as beliefs converge faster. \cite{GolubJackson2010} identified conditions for societal wisdom under DeGroot learning and found {logical independence} between wisdom and convergence speed. Our framework complements this literature by demonstrating that a third factor—degree dependence—affects both wisdom and convergence speed, creating a trade-off between the two.  

Most relatedly, \cite{GolubJackson2012} examined how network features affect convergence speed in large random networks, finding that homophily slows belief convergence. They derived the large-network property of the second-largest eigenvalue in modulus (SLEM), which is inversely related to convergence speed, for equal-weighting learning matrices. We extend their analysis to degree-dependent weighting matrices.

\cite{dasaratha2020distributions} provided asymptotic characterizations of eigenvector and Katz-Bonacich centralities, while \cite{mostagir2021centrality} extended this analysis to other centrality measures, showing that, with high probability, these measures approximate their values in a well-defined ``average'' network. Unlike their focus on the centrality measures for the adjacency matrix, our work investigates the asymptotic properties of the consensus belief and the second-largest eigenvalue in modulus (SLEM) of the learning matrix, which includes nonlinear components arising from degree-dependent weighting.

\cite{cerreia2024dynamic} extended the analysis to a broad class of non-Bayesian belief-updating rules, with DeGroot learning as a special case, and characterized conditions for convergence, consensus, and the wisdom of crowds. They showed that society attained wisdom when the SLEM fell below a certain threshold. While their focus was on generalizing DeGroot learning and identifying conditions for wisdom in that setting, our analysis instead studies the large-sample behavior of the SLEM and the role of degree-dependent weighting.

Last but not least, our study is related to a large body of mathematical literature that examines the convergence and mixing time of Markov chains. We discuss our relation to this literature alongside our technical discussion in the paper. Most interestingly, \cite{sinatra2011maximal} found that constructing maximal-entropy random walks requires step probabilities to be proportional to a power of the degree of the target node, which is similar to the degree-dependent weighting in our paper.  

\section{Model: Networks and Degree-Weighted Learning}
\label{sec:setup}

Consider a network with $n$ agents. The associated $n\times n$ adjacency matrix $\matr{A}$ is defined by: $A_{ij} = A_{ji} = 1$, if there is a link between agents $i$ and $j$; otherwise, $A_{ij}=A_{ji}=0$.\footnote{This suggests that the network is undirected. We discuss the implications of this assumption, as well as the directed case, in Section \ref{sec:discussion}.} Let $d_i(\matr{A}) = \sum_j A_{ij}$ be the degree of agent $i$, i.e., the number of links (friends) $i$ has.

\paragraph{DeGroot Learning Model} In the classical \cite{degroot1974reaching} learning model, an agent's belief such as the probability of global warming or the safety of a vaccine, is a simple linear aggregation of her neighbors' beliefs in the last period. In other words, the belief vector $\matr{b}(t)$ at time $t\in 	\mathbb{N}=\{0,1,\ldots\}$ is given by:
$$\vec{b}(t) = \matr{T}^{t}\vec{b}_0\,,$$ 
in which $\vec{b}_0\in \mathbb{R}^n$ is an initial belief vector and $\matr{T}$, the learning matrix, is an $n \times n$ row-stochastic matrix, with $T_{ij}$ representing the weight that agent $i$ puts on the belief of agent $j$.

\paragraph{Equal-Weighting Benchmark} Typically, agents are assumed to assign equal weight to their neighbors (e.g., \citealp{GolubJackson2010, GolubJackson2012}). Under this assumption, the learning matrix $\matr{T}$ takes the form 
\(
T_{ij} = {A_{ij}}/{d_i(\matr{A})}.
\)

\paragraph{Degree-Weighted Updating} While equal-weighting provides a convenient and tractable benchmark, it falls short of capturing certain real-life patterns discussed in the introduction. These observations suggest a learning heuristic in which the updating process depends on the popularity, or degree, of each source.

To capture these phenomena, we  introduce a degree-weighted learning matrix, \textbf{T}, as follows:
\begin{equation}
T_{ij} = \frac{A_{ij}d_j^\alpha}{\sum\limits_{j'} A_{ij'}d_{j'}^\alpha}\label{NewT}\,,
\end{equation}
where $d_j = d_j(\matr{A})$ is agent $j$'s degree in $\matr{A}$ and $\alpha \in \mathbb{R}$ captures how a neighbor's weight depends on their degree.%
\footnote{We note two points here. First, we allow for more general functional forms of degree dependence, $\phi(\alpha, d) \colon \mathbb{R}^2 \rightarrow \mathbb{R}$, as long as they satisfy certain regularity conditions. We refer readers to  Appendix \ref{subsec:generalized_degree_dependence} for a more detailed discussion. Second, we acknowledge empirical evidence (e.g., \citealp{breza2018seeing}) suggesting that agents may have limited knowledge of network structures and, in particular, may not know their neighbors' degrees exactly. We leave it to future work to explore the implications of this type of uncertainty on social learning.} We have the following cases: 
\begin{itemize}
    \item $\alpha = 0$ (equal-weighting): the benchmark case studied in the literature;
    \item $\alpha > 0$ (popularity-favoring): weights increase with a neighbor’s degree;
    \item $\alpha < 0$ (popularity-resistant): weights decrease with a neighbor's degree.
\end{itemize}

\paragraph{Random Networks: Stochastic Block Model}
Although DeGroot learning with deterministic networks has been extensively studied in the literature, its properties under random networks are less well understood. Random networks, however, offer distinct appeal and advantages over deterministic ones.\footnote{For example, researchers often face data limitations and rely on statistical models of network formation, where links are probabilistically determined (see \citealp{dasaratha2020distributions} for more examples). Moreover, in real-world settings, agents may encounter uncertainty regarding whom they will be matched with, which is better modeled using random networks (as in \citealp{acemoglu2010spread}).} We use the stochastic block model to generate random networks in our paper.\footnote{See \cite{Holland1983} for a classical discussion and \cite{lee2019review} for a recent survey on its development.}

\begin{figure}
\includegraphics[width=80mm]{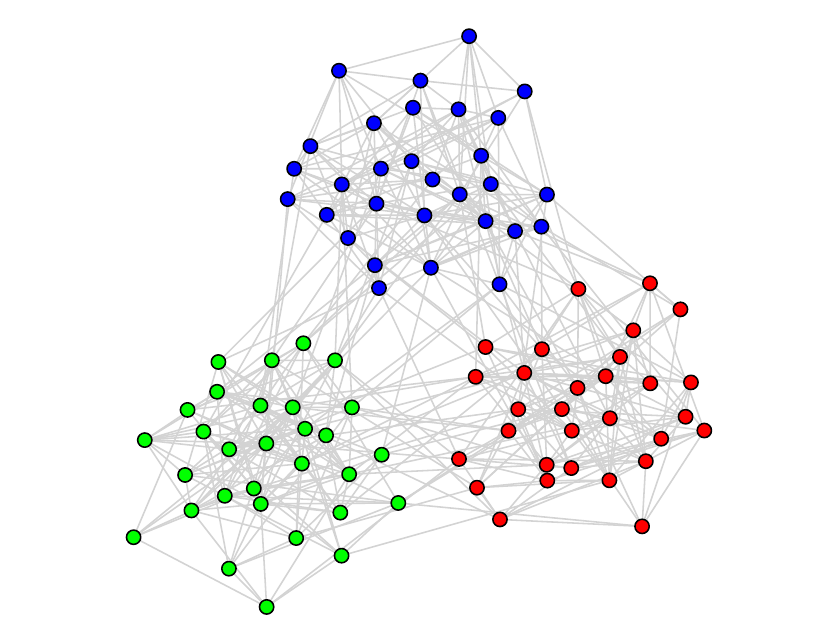}
\caption{A Network Generated by Stochastic Block Models} 

\begin{figurenotes}
A network generated by a stochastic block model with $3$ groups, $30$ agents in each. The linking probability between agents $i$ and $j$ is given by
$p_{ij} = 0.3$ if $i$ and $j$ belong to the same group, and $0.04$ otherwise.
\end{figurenotes}
\end{figure}


Suppose there are $n$ agents divided into $m$ groups. Let $\vec{n} = (n_1,\ldots,n_m)$, with $n_k$ being the number of agents in group $k\in \{1,\ldots,m\}$. 
Let \textbf{P} be an $m \times m$ symmetric matrix
\begin{equation}
\label{eq:P}
\textbf{P} = \begin{bmatrix}
p_{11} & p_{12} & \ldots & p_{1m}\\
 p_{21} & p_{22} & \ldots & p_{2m}\\
\vdots & \vdots &\ddots &\vdots\\
p_{m1} & p_{m2} & \ldots & p_{mm}
\end{bmatrix}\,,
\end{equation}
where {$p_{ij}\geq 0$} represents the linking probability between any agent in group $i$ and any agent in group $j$.

Let $\matr{A}(\matr{P}, \matr{n})$ denote a random matrix generated by the corresponding stochastic block model. The entry $A_{ij}$ represents the element in the $i$-th row and $j$-th column of $\matr{A}$. Each $A_{ij}$ is an independent Bernoulli random variable, where $A_{ij} \sim \mathrm{Bernoulli}(P_{kl})$, for $i \neq j$, and agents $i$ and $j$ belong to groups $k$ and $l$, respectively.

Let $\matr{R}=(R_{ij})=\E \matr{A}$ be the corresponding {\em expected adjacency matrix}. 
The {\em expected degree} of a node $i$ in some group $k_i$ is
$$d_i(\matr{R}) = \sum_{l=1,\ldots,m} p_{k_i l} n_l.
$$

Together with any $\alpha$, we define the {\em expected learning matrix}, $\matr{T}^*(\matr{R},\alpha)$,\footnote{We note here that the role of matrices \(\matr{R}\) and \(\matr{T}^*\) is purely instrumental. They do not represent the actual network or the updating process but instead provide a tractable approximation that facilitates the analysis of the expectation and limiting behavior of the opinion dynamics.} 
\begin{equation}\label{NewT_R}
T^*_{ij} = \frac{R_{ij}d_j^\alpha(\matr{R})}{\sum\limits_{j'}
R_{ij'}d_{j'}^\alpha(\matr{R})}\,, \forall i,j.
\end{equation}



We adopt standard definitions of convergence, consensus, and convergence speed as follows:


\begin{definition}[Belief Convergence and Consensus]  
Beliefs are said to \textit{converge} if, for any initial belief profile \( b(0) \), the sequence \( b(t) \) satisfies \(\lim_{t \rightarrow \infty} b_i(t)\) exists for all \( i \).  
The network reaches \textit{consensus} if beliefs converge to a common value, that is, \(\lim_{t \rightarrow \infty} b_i(t) = b^*\) for all \( i \), where \( b^* \) is a constant independent of \( i \).  
\end{definition}

\begin{definition}[Convergence Speed]
\label{def:convergence_speed}
Denote \( \matr{T}^{\infty} = \lim_{t\rightarrow\infty}\matr{T}^{t} \). 
Then, at time \( t \), the convergence speed is measured by
\begin{equation}
\max\limits_{\|\vec{b}_0\|_{\matr{s}}=1} \| (\matr{T}^t - \matr{T}^{\infty}) \vec{b}_0 \|_{\matr{s}},
\end{equation}
in which 
$
\|\mathbf{v}\|_{\mathbf{s}}=\sqrt{\sum_{i=1}^n s_i\mathbf{v}_i^2},
$
where $\mathbf{v}_i$ is the $i$-th entry of $\mathbf{v}$ and $s_i$ is agent $i$'s influence to be defined in Equation \eqref{eq_influence}.
\end{definition}

This expression quantifies the largest possible deviation between beliefs at time \( t \) and the limiting belief, across all initial beliefs whose $\|\cdot\|_{\matr{s}}$ norm is standardized to 1.\footnote{The constraint $\|\mathbf{b}_0\|_{\matr{s}} = 1$ is a standardization. Any initial belief $\mathbf{b}_0$ can be expressed as $\mathbf{b}_0 = \|\mathbf{b}_0\|_{\matr{s}} \cdot \mathbf{b}_0 / \|\mathbf{b}_0\|_{\matr{s}}$. Therefore, this effectively allows for any choice of initial belief.} A larger value indicates slower convergence.\footnote{The use of the $\|\cdot\|_{\matr{s}}$ norm is inspired by a similar norm in \cite{GolubJackson2012}. It places greater weight on discrepancies in the positions of more influential agents relative to the limit. We adopt the $\|\cdot\|_{\matr{s}}$ norm because, in the ${\matr{s}}$-weighted $L^2$ space, the operator $\matr{T}$ is self-adjoint and its eigenvectors are orthogonal, allowing us to apply standard techniques as in \citet{diaconis1991geometric} to prove the second part of Proposition \ref{prop_expected}.}

\begin{definition}A network is  \emph{strongly connected} if there exists a path connecting any two nodes of the network. A network is \emph{aperiodic} if there is no integer $k>1$ dividing the length of every cycle of the network. A nonnegative square matrix $\matr{M}$ is \emph{primitive} if there exists a positive integer $k$ such that $\matr{M}^k$ is strictly positive (i.e., all entries of $\matr{M}^k$ are greater than zero).   
\end{definition}

It is a standard result that a strongly connected network is primitive if and only if it is aperiodic (see, e.g., Theorem~1.4 in \citealp{seneta2006non}).

Lastly, we impose the following assumptions to establish our asymptotic results. We let the number of groups, $m$, be fixed as $n \rightarrow \infty$.


\begin{assumption}[Density]
\label{ass:density}
\(
\lim_{n\rightarrow \infty}{\tau_n}/{\sqrt{{(\log n)}/{n}}} = \infty \,,
\) where $\tau_n = \min_i {d_{i}(\matr{R})}/{n}$.
\end{assumption}

 \begin{assumption}[{Comparable degrees}]
\label{ass:cdensities}
{
\(
{\max_{1\le i,j\le m}\{p_{ij}\}}/{\tau_n}<\infty\,.
\)}
\end{assumption}

\begin{assumption}[Bounded spectral gap]
\label{ass:connection}
   $\limsup_{n\rightarrow\infty}|\lambda_2(\matr{T}^*)|<1\,.$
\end{assumption}

Assumption~\ref{ass:density} requires that the graph is not too sparse; in particular, the minimal degree must be significantly larger than $\sqrt{n \log n}$. 
Assumption~\ref{ass:cdensities} requires that the degrees in the network are comparable.\footnote{Assumption~\ref{ass:cdensities} aligns with Condition~(4) in Definition~3 of \cite{GolubJackson2012}.
Assumption~\ref{ass:density} differs from Condition~(1) in Definition~3 of \cite{GolubJackson2012}, as it is specifically needed to prove concentration of nonlinear terms, such as $d_j^\alpha$.}
Assumption~\ref{ass:connection} is similar to Condition~(3) in Definition~3 of \cite{GolubJackson2012},  ensuring that $\matr{A}$ is primitive with high probability and that $\matr{R}$ is primitive for sufficiently large $n$.\footnote{See Appendix \ref{sec:Tlimit} for a formal statement and proof. Because $\lambda_2$ determines convergence speed (see Proposition~\ref{prop_expected}), this assumption prevents convergence from becoming arbitrarily slow as the network grows. It is stated as a condition on eigenvalues; we provide alternative sufficient conditions based solely on model primitives in Online Appendix O4.
}
All three assumptions ensure that, as networks grow, our characterization of learning outcomes in the deterministic case (as captured by Proposition~\ref{prop_expected}) holds with high probability, and that learning outcomes in realized networks are sufficiently close to those in expected networks (Theorem~\ref{thm:asy}).


\section{General Results}
\label{sec:result}

In this section, we present our main theoretical result for the degree-weighted updating rule: an asymptotic theorem showing that, as the network grows large, the learning outcomes—specifically, the consensus belief and the convergence speed of the realized learning matrix $\matr{T}(\matr{A};\alpha)$—converge to those of the expected learning matrix $\matr{T}^*(\matr{R};\alpha)$. This result allows us to focus on the expected case, $\matr{T}^*(\matr{R};\alpha)$, when discussing the comparative statics of learning outcomes in Section~\ref{sec:app}.  

As preparation, we begin by characterizing consensus beliefs and convergence speed for deterministic networks, with the expected network as a special case. The first part of Proposition \ref{prop_expected} adapts the standard characterization of consensus beliefs in DeGroot learning to our degree-weighted setting, and the second part generalizes Lemma 2 of \cite{GolubJackson2012}.

\begin{proposition}[Learning Outcomes for Deterministic and Expected Matrices]\label{prop_expected}
Given any primitive matrix $\matr{A}$ such that $\min_i d_i(\matr{A})>0$,\footnote{The expected adjacency matrix \( \matr{R} = \E \matr{A} \), generated by the stochastic block model, is primitive under Assumptions~\ref{ass:density}–\ref{ass:connection} (see Proposition~\ref{remark_connect} in the Appendix). Therefore, in this proposition we can replace \( \matr{A} \) with \( \matr{R} \) and \( \matr{T} \) with \( \matr{T}^* \). Assumptions~\ref{ass:density}–\ref{ass:connection} also imply that \( \matr{A} \) is primitive with probability tending to one, which in turn implies that the precondition for Proposition~\ref{prop_expected} holds with high probability.}  
we have:
\begin{enumerate}
\item Beliefs converge to a consensus \(b(\infty)=(s_1,\ldots,s_n)\vec{b}_0\), where agent \(i\)'s influence \(s_i\) is
\begin{equation}\label{eq_influence}
s_{i}=\frac{\sum\limits_j A_{ij}\, d_{j}^{\alpha}(\matr{A})\, d_i^{\alpha}(\matr{A})}{\sum\limits_{i',j} A_{i'j}\, d_j^{\alpha}(\matr{A})\, d_{i'}^{\alpha}(\matr{A})}.
\end{equation}
\item Let \(|\lambda_2(\matr{T})|\) denote the second-largest eigenvalue (in modulus) of \(\matr{T}\) and let \(\|\mathbf{v}\|_{\mathbf{s}}=\sqrt{\sum_{i=1}^ns_i\mathbf{v}_i^2}\), where $\mathbf{v}_i$ is the $i$-th entry of $\mathbf{v}$ and $s_i$ is agent $i$'s influence defined in Equation \eqref{eq_influence}. Then, at any time \( t \),
\begin{equation}\label{eq:speed_general}
\max\limits_{\|\vec{b}_0 \|_{\mathbf{s}}=1} \| (\matr{T}^t - \matr{T}^{\infty}) \vec{b}_0 \|_{\mathbf{s}} = |\lambda_2(\matr{T})|^t.
\end{equation}
Therefore, the convergence speed is inversely related to \(|\lambda_2(\matr{T})|\).

\end{enumerate}

\end{proposition}

We note a few things here. It is a standard result in the literature that if a network is primitive, long-run beliefs converge to consensus,\footnote{See \cite{degroot1974reaching} and \cite{Meyer2000}; specifically, Proposition~6 in \cite{golubs2016}.} and that the convergence speed of beliefs is determined by the second-largest eigenvalue in modulus (SLEM) of the learning matrix.\footnote{See \cite{montenegro2006mathematical} and \cite{levin2017markov}; see also Proposition~8 in \cite{golubs2016}.} Our characterization of agents' influences and convergence speed in the degree-weighted deterministic setting (Proposition \ref{prop_expected}) provides the foundation for the asymptotic result below and for applying our framework to a wisdom-of-crowds result and the comparative statics in Section~\ref{sec:app}. 




We are now ready to introduce our key technical result. We show that the learning outcomes corresponding to the realized networks, $\matr{A}$ and $\matr{T}$, are, in a probabilistic sense, sufficiently close to those corresponding to the expected networks, $\matr{R}$ and $\matr{T}^*$.\footnote{In Online Appendix~O3, 
we show that our main results are robust to minor perturbations of the adjacency matrix \( \matr{A} \).}



\begin{theorem}\label{thm:asy}
Under Assumptions~\ref{ass:density}–\ref{ass:connection}, for any degree dependence $\alpha \in \mathbb{R}$ and any random network $\vec{A}(\vec{P},\vec{n})$, there exists a positive constant $\tilde{C}$, independent of $n$, such that for all $n>0$, the following results hold for the realized learning matrix $\matr{T}(\matr{A},\alpha)$ and the expected learning matrix $\matr{T}^* = \matr{T}(\matr{R},\alpha)$.
\begin{enumerate}
\item \textbf{For influences,}\footnote{Recall that $\tau_n = \min_i {d_{i}(\matr{R})}/{n}$. Assumption \ref{ass:density} implies that ${\sqrt{\log n}}/{(\tau_n\sqrt{n})}$ goes to $0$ as $n \rightarrow \infty$.}
$$
\p\Bigl\{
\|\vec{s}(\matr{T})-\vec{s}(\matr{T}^*)\|_2\ge \tilde{C}\frac{\sqrt{\log n}}{\tau_n\sqrt{n}}
\Bigr\}\le \frac{16}{n^2}\,.$$%

\item \textbf{For convergence speed,}
$$
\p\Bigl\{
\bigl|\lambda_2(\matr{T})-\lambda_2(\matr{T}^*)\bigr|\ge \tilde{C}\frac{\sqrt{\log n}}{\tau_n\sqrt{n}}
\Bigr\}\le \frac{16}{n^2}\,.
$$
\end{enumerate}
\end{theorem}
\label{proof_s1}
\paragraph{Sketch of Proof}\footnote{The formal proof is involved and is provided in  Appendix \ref{sec:proof_thm_asy}. We briefly outline the key ideas here and indicate how our approach diverges from the existing literature.} 
The critical step of the proof is to show that the symmetry of the adjacency matrix \( \matr{A} \) ensures that the learning matrices \( \matr{T} \) and \( \matr{T}^* \) are similar to symmetric matrices. This property is pivotal, as it enables the application of eigenvalue perturbation methods specifically designed for symmetric or symmetrically similar matrices.

Specifically, we express \(
    \matr{T} = \matr{D}_{1,\matr{A}}^{-1} \matr{A} \matr{D}_{2,\matr{A}},\) and
    \(\matr{T}^* = \matr{D}_{1,\matr{R}}^{-1} \matr{R} \matr{D}_{2,\matr{R}}\) 
where \( \matr{D}_{i,\matr{A}} \) and \( \matr{D}_{i,\matr{R}} \) are diagonal matrices that depend nonlinearly on degrees. Lemma~\ref{lemma:similar} in the appendix shows how the symmetry of the adjacency matrix \( \matr{A} \) implies that \( \matr{T} \) and \( \matr{T}^* \) are similar to symmetric matrices.

With symmetry established, we then bound the difference \( |\lambda_2(\matr{T})-\lambda_2(\matr{T}^*)| \) through Weyl's inequality, a standard perturbation result that requires symmetry (or symmetric similarity) of the involved matrices. This reduces our analysis to deriving suitable bounds on the matrix differences \( \|\matr{A} - \matr{R}\| \), \( \|\matr{D}_{1,\matr{A}} - \matr{D}_{1,\matr{R}}\| \), and \( \|\matr{D}_{2,\matr{A}} - \matr{D}_{2,\matr{R}}\| \).\footnote{{For a matrix $\matr{B}$, $\|\matr{B}\|=\sup_{\|\matr{v}\|=1}=\|\matr{B}\matr{v}\|$, where $\|\matr{v}\|$ is Euclidean norm of vector $\matr{v}$.}}

Bounding \( \|\matr{A} - \matr{R}\| \) is relatively straightforward and follows standard procedures from the literature (e.g., \citealp{GolubJackson2012,dasaratha2020distributions,mostagir2021centrality}). However, the nonlinearity introduced by degree weighting presents significant difficulties in bounding the diagonal matrices \( \matr{D}_{1,\matr{A}} \) and \( \matr{D}_{2,\matr{A}} \).

To handle this nonlinearity, we employ a leave-one-out strategy, inspired by techniques such as those developed by \cite{abbe2020entrywise}. This approach involves replacing \( d_j^{\alpha}(\matr{A}) \) with \( d_{j,-i}^{\alpha}(\matr{A}) \), which excludes the direct influence of the element \( A_{ij} \). Consequently, the modified degree \( d_{j,-i}^{\alpha}(\matr{A}) \) becomes independent of \( A_{ij} \), allowing us to approximate the diagonal entries as
\(
(\matr{D}_{1,\matr{A}})_{ii} \approx \sum_j A_{ij}\,d_{j,-i}^{\alpha}(\matr{A}).
\)

Given this construction, the sum \(\sum_{j\neq i} A_{ij}\,d_{j,-i}^{\alpha}(\matr{A})\) becomes a sum of independent random variables, conditional on \(\{d_{j,-i}(\matr{A}): j\neq i\}\). This step allows us to employ standard concentration inequalities such as Bernstein inequality \citep{vershynin2018}  to control the convergence rates of the diagonal entries of \( \matr{D}_{1,\matr{A}} \) and \( \matr{D}_{2,\matr{A}} \). This concludes the key steps of the proof.\footnote{{We briefly outline the intuition for each assumption in the context of this theorem. The proof relies on bounding the deviation between the eigenvectors of the random matrix $\matr{T}$ and the expected matrix $\matr{T}^*$. This bound generally takes the form $\text{Error} \lesssim {\text{Noise}}/{\text{Gap}}$. Assumptions~\ref{ass:density} and~\ref{ass:cdensities} control the noise (numerator): they ensure sufficiently dense connections so that the random matrix concentrates tightly around its mean via the Law of Large Numbers. Assumption~\ref{ass:connection} guarantees the gap (denominator): it ensures the spectral gap of the expected matrix remains strictly bounded away from zero, thereby providing the structural stability required for the eigenvectors to be robust to noise.}} \hfill \qed
\label{proof_s2}

\begin{remark}
\label{rm:alpha}
Proposition~\ref{prop_expected} and Theorem~\ref{thm:asy} remain valid when \( \alpha \) depends on the individual \( i \), i.e.,
\[
T_{ij} = \frac{A_{ij}d_j^{\alpha_j}}{\sum\limits_{j'} A_{ij'}d_{j'}^{\alpha_{j'}}},\quad
T^*_{ij} = \frac{R_{ij}d_j^{\alpha_j}(\matr{R})}{\sum\limits_{j'} R_{ij'}d_{j'}^{\alpha_{j'}}(\matr{R})},\quad
s_{i} = \frac{\sum\limits_j R_{ij}\, d_{j}^{\alpha_j}(\matr{R}) d_i^{\alpha_i}(\matr{R})}{\sum\limits_{i',j}R_{i'j} d_j^{\alpha_j}(\matr{R}) d_{i'}^{\alpha_{i'}}(\matr{R})},\quad \forall i,j,
\]
where \( \max_{i,j}|\alpha_i - \alpha_j| \le c \) for some positive constant \( c \).\footnote{See Appendix \ref{proof:rm_alpha} for a proof.}
\end{remark}


Theorem~\ref{thm:asy} lays the foundation for further discussion of economic insights in the degree‐weighted learning framework. In particular, to understand key learning outcomes, we are able to focus on the expected matrices rather than all possible realized networks. Together with Proposition~\ref{prop_expected} in the previous section, we know that influences (and later, the wisdom) are captured by Equation~(\ref{eq_influence}), and that the convergence speed is inversely related to $|\lambda_2|$ of the expected learning matrix $\matr{T}^*$, as defined in Equation~(\ref{eq:speed_general}). In the next section, we provide a series of examples to highlight the impacts of degree‐weighted learning, as captured by $\alpha$, on societal wisdom, convergence speed, and a potential trade-off between the two.

\section{Applications}
\label{sec:app}
In this section, we apply the asymptotic results established above to three questions. First, we examine how the degree-dependent learning rule affects societal wisdom and find that its impact is largely monotonic. Next, we study its effect on the speed of learning, which displays a non-monotonic pattern. Lastly, we discuss a potential trade-off between societal wisdom and learning speed. With Theorem~\ref{thm:asy}, we can focus on the comparative statics of expected networks rather than of specific realized networks.

\paragraph{A Two-Type Model} For analytical tractability, we consider a class of stochastic block models with two levels of expected degree. Specifically, suppose there are $m$ groups with sizes $n_1 \neq n_2 = \cdots = n_m$, where within-group linking probabilities are given by $p_{kk} = p$ and across-group probabilities by $p_{kl} = q$ for $l \neq k$, with $p \neq q$.
 Agents therefore fall into two types: group~1 has expected degree $d_1^* = n_1 p + (m-1) n_2 q$, and each group $k > 1$ has expected degree $d_2^* = n_1 q + n_2 p + (m-2) n_2 q$. We refer to the group(s) with larger expected degree as the high-degree group(s), and those with smaller expected degree as the low-degree group(s).\footnote{Group~1 is the high-degree group whenever $d_1^* - d_2^* = (n_1 - n_2)(p - q) > 0$.}

{This structure, though simple, accommodates three cases distinguished by the numbers of high-degree ($m_h$) and low-degree ($m_l$) groups: (i) $m_h = 1$ and $m_l = 1$; (ii) $m_h > 1$ and $m_l = 1$, capturing situations in which high-degree agents represent different campaigns or hold conflicting views; and (iii) $m_h = 1$ and $m_l > 1$, corresponding to environments with multiple low-degree groups that may reflect cultural, ethnic, or ideological divisions.
}


\subsection{Influence and Wisdom of the Crowd}
\label{sec:wisdom}
This section discusses influence and wisdom.\footnote{We thank the editor and the anonymous referees for suggesting this discussion.} {Our goal is to combine insights from \cite{GolubJackson2010} with our main technical result, Theorem~\ref{thm:asy}, to characterize the effect of degree dependence on societal wisdom.}

{Specifically, we consider a sequence of realized networks generated by a stochastic block model, and analyze the corresponding sequence of learning matrices. The wisdom criterion, established in \cite{GolubJackson2010}, states that society is wise if and only if $\lim\limits_{n \rightarrow \infty} \max\limits_{i \leq n} s_i(\matr{T}(n)) = 0$, where $s_i(\matr{T}(n))$ is agent $i$'s influence under the learning matrix $\matr{T}(n)$. Our Theorem~\ref{thm:asy} shows that, under Assumptions~\ref{ass:density}–\ref{ass:connection}, the influence vector of the realized network converges to that of the expected network as $n \rightarrow \infty$. Therefore, if $\lim\limits_{n \rightarrow \infty} \max\limits_{i \leq n} s_i(\matr{T}^*(n)) = 0$ for the expected matrix $\matr{T}^*(n)$, then, with high probability, the realized sequence of networks will also be wise. Conversely, if $\liminf\limits_{n} \max\limits_{i \leq n} s_i(\matr{T}^*(n)) > 0$, wisdom fails with high probability. This logic ensures that our comparative static results with respect to $\alpha$ hold with high probability in large networks.}


\paragraph{Influences With Special Cases of $\alpha$} We begin by considering special cases of the degree-dependence parameter~$\alpha$ to illustrate its effect on agents' influence. 
Recall that in Equation~(\ref{eq_influence}) of Proposition~\ref{prop_expected}, we obtain the following characterization of the impact of $\alpha$ on influence for any deterministic matrix $\matr{A}$ that satisfies the imposed conditions:  
$s_{i} = {(\sum_j A_{ij}\, d_{j}^{\alpha} d_i^{\alpha})}/{(\sum_{i',j} A_{i'j}\, d_j^{\alpha} d_{i'}^{\alpha})}.
$

A notable special case is $\alpha = 0$ (equal weighting). As shown in the literature (e.g., \citealp{demarzo2003persuasion} and \citealp{GolubJackson2010}), we have
$s_i = {d_i}/{\sum_{j} d_j}$,
so that the influence of agent $i$ is proportional to their degree.

Our framework introduces additional interesting cases. For instance, the following example presents the case of $\alpha = -1$:

\begin{example}
When $\alpha = -1$, an agent $i$'s influence is 
$s_i = {\sum_j ({A_{ij}}/{d_i d_j}})/{\sum_{i',j}({A_{i'j}}/{d_{i'} d_j}}),$
which is inversely proportional to the (weighted) harmonic average of the neighbors' degrees. When the neighbors' degrees are independent, all agents share the same influence in expectation.
\end{example}

This specific degree dependence, $\alpha = -1$, approximately offsets the impact of degree heterogeneity, though it does not yield exactly equal influences.

Next, we highlight that the degree dependence, $\alpha$, via shaping the agents' influences, can have more fundamental impacts on learning outcomes, such as whether the social learning process results in an unbiased aggregation of individual signals (i.e., the wisdom of the crowd).

\paragraph{Wisdom} 
We follow the definition of wisdom introduced in \cite{GolubJackson2010}. More specifically, suppose that agents' initial beliefs are i.i.d.\ signals drawn from a distribution with mean $\theta$, which represents the true state of the world, and variance $\sigma^2>0$. To study networks in the large, we index the learning matrix by $n$, representing the size of the society. The sequence of learning matrices is then given by a sequence of $n$-by-$n$ interaction matrices. We denote the belief of agent $i$ in network $n$ at time $t$ by $b_i^{(t)}(n)$ and its limiting belief by $b_i^{(\infty)}(n)$. We say that a sequence $\bigl(\matr{T}(n)\bigr)_{n=1}^\infty$ of learning matrices is \emph{wise} if the limiting beliefs converge jointly in probability to the true state $\theta$.

\begin{definition}[Wisdom]
The sequence $\bigl(\matr{T}(n)\bigr)_{n=1}^\infty$ is \emph{wise} if 
$\plim\limits_{n\rightarrow\infty} \; \max_{i \leq n} \bigl| b_i^{(\infty)}(n) - \theta \bigr| = 0.$
\end{definition}

We also say that a society is wise if the sequence $\bigl(\matr{T}(n)\bigr)_{n=1}^\infty$ is wise. Let $s(n)$ denote the influence vector corresponding to $\matr{T}(n)$, as defined in Proposition~\ref{prop_expected}. By applying Lemma 1 and Corollary 1 in \cite{GolubJackson2010} (with our formulation for $s(n)$ replacing their $s(n)_1$), we obtain an explicit necessary and sufficient condition for a society to be wise when agents use a degree‐dependent updating rule.

\begin{lemma}[\citealp{GolubJackson2010}]
\label{cor:influence}
A society is wise if and only if 
$\lim\limits_{n\rightarrow\infty}\max_{i} s_i(n) = 0.$
\end{lemma}


 

 With the above preparations, we are now ready to state our results on the effect of degree dependence, $\alpha$, on influence and societal wisdom. Recall that for any expected network $\matr{R}$ generated by a two-type model introduced at the beginning of this section, there are two degree levels, denoted by $d_l < d_h$. Let the corresponding counts of agents be $n_h$ and $n_l = n - n_h$, and the influences of each type be $s_l$ and $s_h$, respectively. We have the following:

\begin{proposition}
[Monotonic Impact of \( \alpha \) on Influences] In the two-type model, as $\alpha$ increases, $s_h(\alpha)$ increases and $s_l(\alpha)$ decreases. 
In addition, at the extremes, we have $s_h(-\infty) := \lim_{\alpha \rightarrow  -\infty} s_h(\alpha)= 0$, $s_l(-\infty) = 1/n_l$, and $s_h(+\infty) = 1/n_h$, $s_l(+\infty) = 0$.
\label{prop-influence} 
\end{proposition}

That is, assigning more weight to high-degree neighbors amplifies the influence of high-degree agents and diminishes the influence of low-degree agents. At the extremes, when $\alpha \rightarrow -\infty$, high-degree neighbors are effectively ignored so that each low-degree agent carries equal influence, while the influence of any high-degree agent vanishes. Conversely, when $\alpha \rightarrow +\infty$, weight is placed almost exclusively on high-degree neighbors, so that each high-degree agent carries equal influence, while the influence of any low-degree agent vanishes. \\

With Lemma~\ref{cor:influence} and Proposition~\ref{prop-influence}, we now state the comparative statics on wisdom.

\begin{proposition}
[Monotonic Impact of \( \alpha \) on Wisdom]
\label{cor-wisdom} 
In the two-type model, consider a sequence of societies indexed by the number of agents $n$, such that the number of high-degree agents is bounded, i.e., $n_h(n) < M$ for some $M < +\infty$. Then there exists a threshold $\bar{\alpha}$ such that the learning outcome is wise if and only if $\alpha \leq \bar{\alpha}$. At the extremes, the learning outcome is wise when $\alpha = -\infty$ and unwise when $\alpha = +\infty$.
\end{proposition}



\paragraph{Key Takeaway on Societal Wisdom} In networks with a bounded number of high-degree agents and infinitely many low-degree agents, degree dependence \( \alpha \) has \emph{monotonic} effects on influence and thus on societal wisdom: higher values of \( \alpha \) enable high-degree agents to retain non-vanishing influence, rendering the society unwise.

Wisdom requires that every individual’s influence vanish in the limit, whereas non-wisdom arises when a change in \( \alpha \) grants a certain group non-vanishing influence. In our setting, as \( \alpha \) increases, high-degree agents—who constitute a bounded group such as experts or media—may preserve influence, thereby rendering the society unwise. This reflects a structural asymmetry: while the group of low-degree agents grows unboundedly and their influence dissipates, the group of high-degree agents remains bounded and thus can sustain non-trivial influence in large societies (see Proposition \ref{prop-influence}).

\subsection{Convergence Speed} 
\label{sec:speed}
{In this section, we illustrate the effect of the degree-dependence \( \alpha \) on convergence speed. The analysis builds on Proposition~\ref{prop_expected}, which establishes that convergence speed is inversely related to \( |\lambda_2| \) of the expected network, together with Theorem~\ref{thm:asy}, which ensures that the comparative statics of $\alpha$ on convergence speed hold with high probability in large networks.} 

We show that, in contrast to the monotonic effect of $\alpha$ on wisdom, its impact on convergence speed is non-monotonic. Whether a more popularity-favoring rule (higher \( \alpha \)) accelerates or slows convergence depends on the relative number of groups, which essentially determines the diversity of views potentially held within each type. 

\begin{figure}[h]
\includegraphics[width=100mm]{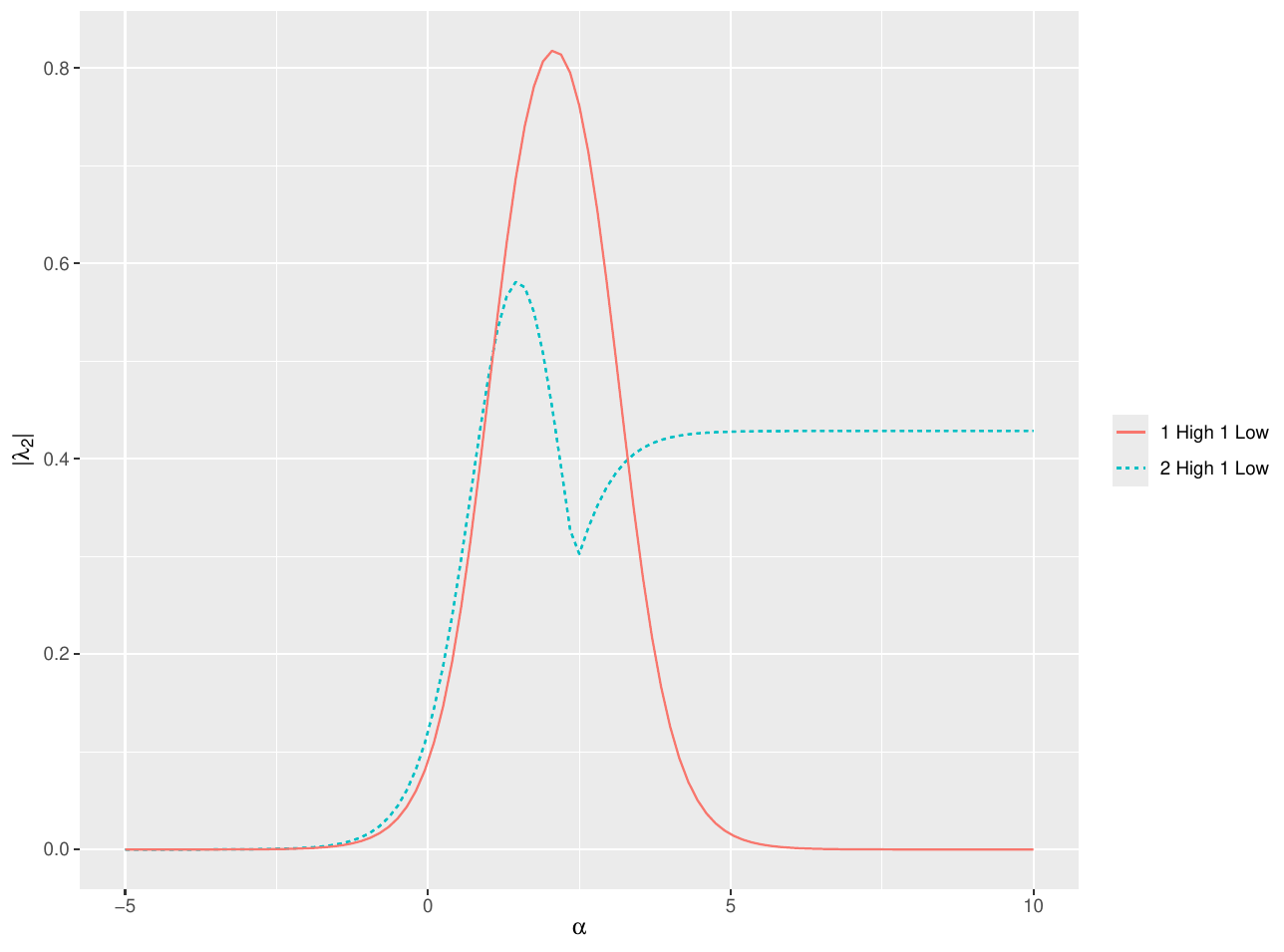}
\caption{Convergence Speed}
\label{fig_sym_speed_note}

\begin{figurenotes}
Convergence speed, as measured by \( |\lambda_2(\mathbf{T}^*)| \), is plotted against degree dependence \( \alpha \). The solid curve corresponds to the case with one high-degree group and one low-degree group (Example~\ref{ex-1-e-1-g}), while the dashed curve corresponds to the case with two high-degree groups and one low-degree group (Example~\ref{ex-2-e-1-g}).\end{figurenotes}
\end{figure}

We begin with the following example.


\begin{example}[One High-Degree Group and One Low-Degree Group]
\label{ex-1-e-1-g}
Consider a network consisting of one low-degree group (of size $n_1 = 1000$) and one high-degree group (of size $n_2 = 10$), where the within-group linking probabilities are \( p_{11} = p_{22} = 0.01 \), and the across-group linking probabilities are \( p_{12} = p_{21} = 0.10 \).
 There exists a threshold \( \alpha^* \) at which convergence is slowest. As \( \alpha \) moves away from this threshold, $|\lambda_2|$ of the expected learning matrix decreases, implying faster convergence.
\end{example}

As the preceding example illustrates (see the solid curve in Figure~\ref{fig_sym_speed_note}), variations in \( \alpha \) lead to non-monotonic changes in convergence speed. Notably, convergence speed is not simply governed by a uniform bias toward either group; rather, it is determined by the \emph{relative} influence between the high- and low-degree groups. In the setting with one high-degree group and one low-degree group, at most two dominant opinions emerge. When \( \alpha \) is sufficiently large, the high-degree group dominates, resulting in rapid convergence; similarly, as \( \alpha \rightarrow -\infty \), the low-degree group prevails, again yielding fast convergence. The slowest convergence arises for an intermediate value \( \alpha^* \), where the influence of the two groups is most balanced.


The following proposition generalizes Example~\ref{ex-1-e-1-g} and provides an explicit characterization of the convergence speed for networks with two groups—one high-degree and one low-degree.
\begin{proposition}[Convergence Speed with One High-Degree and One Low-Degree Group]
\label{prop-1-e-1-g}
Consider a two-group model with $n_1 > n_2$ and linking probabilities $p_{11} = p_{22} = p$ and $p_{12} = p_{21} = q$. The expected degrees of groups 1 and 2 are $d^*_1 = n_1 p + n_2 q$ and $d^*_2 = n_1 q + n_2 p$, respectively. The convergence speed can then be approximated by the inverse of the following expression:
\begin{equation}
    \label{lambda2m2}
    |\lambda_2(\matr{T}^*)|=\left|
    \frac{n_1p\,d_1^{*\alpha}}{n_1p\,d_1^{*\alpha}+n_2q\,d_2^{*\alpha}}
    -\frac{n_1q\,d_1^{*\alpha}}{n_1q\,d_1^{*\alpha}+n_2p\,d_2^{*\alpha}}
    \right|.
\end{equation}
This expression has the following properties: 
\begin{itemize}
    \item $\lim\limits_{\alpha\rightarrow -\infty}|\lambda_2(\mathbf{T}^*)|= 0$,  $\lim\limits_{\alpha\rightarrow +\infty}|\lambda_2(\mathbf{T}^*)| = 0$.
    \item $|\lambda_2(\matr{T}^*)|$ attains its maximum, $\overline{|\lambda_2|} = {|p-q|}/{(p+q)}$, at $\alpha^* = \log_{d^*_2/d^*_1}(n_1/n_2)$.
    \item $|\lambda_2(\matr{T}^*)|$ increases with $\alpha$ for $\alpha < \alpha^*$ and decreases with $\alpha$ for $\alpha > \alpha^*$.
\end{itemize}
\end{proposition}


The proposition reflects a general pattern of non-monotonicity in the impact of $\alpha$, as illustrated in Example~\ref{ex-1-e-1-g}, and shows that convergence is immediate when $\alpha \rightarrow +\infty$ or $\alpha \rightarrow -\infty$, as a single group of homogeneous agents (in expectation) dominates the learning dynamics.

The explicit characterization of $|\lambda_2|$ also provides new insight beyond the example: as shown in the proposition, the slowest rate of convergence is captured by $\overline{|\lambda_2|} = {|p-q|}/{(p+q)}$. Recall that $p$ denotes the linking probability within groups and $q$ the linking probability across groups. When $p > q$, this recovers the classical result that increased homophily slows convergence \citep{GolubJackson2012}. More interestingly, when $p < q$, the slowest rate of learning decreases in $|p - q|$, indicating that greater group differentiation also slows convergence. Finally, if $p = q$, the groups are indistinguishable and the population is homogeneous (in expectation), yielding immediate convergence for any value of $\alpha$.\\

To further illustrate that it is the diversity of views that drives the non-monotonic impact of $\alpha$, we consider the following example, which splits the single high-degree group in Example~\ref{ex-1-e-1-g} into two groups.
\begin{example}[Two High-Degree Groups and One Low-Degree Group]
\label{ex-2-e-1-g}
Let $n_1 = 1000$ and $n_2 = n_3 = 5$, with linking probabilities \( p_{kk} = 0.01 \) within each group, \( p_{1k} = p_{k1} = 0.10 \) for \( k = 2,3 \) between the high-degree and low-degree groups, and \( p_{23} = p_{32} = 0.001 \) between the two high-degree groups. 
The learning speed for this configuration is plotted as the dashed curve in Figure~\ref{fig_sym_speed_note}. For small values of $\alpha$, it closely resembles Example~\ref{ex-1-e-1-g} (the solid curve). For large positive $\alpha$, however, the pattern diverges markedly: convergence slows down as \( \alpha \rightarrow +\infty \).
\end{example}

Examples~\ref{ex-1-e-1-g} and \ref{ex-2-e-1-g} feature almost identical degree distributions,\footnote{The parameters are chosen so that in both examples, there are 10 high-degree agents with expected degree \( d_h \approx 100 \) and 1000 low-degree agents with expected degree \( d_l \approx 11 \).}
yet exhibit markedly different patterns of convergence speed.
Unlike the single high-degree group case (Example~\ref{ex-1-e-1-g}), with multiple high-degree groups (Example~\ref{ex-2-e-1-g}), diverse views may arise among them; consequently, as \( \alpha \rightarrow +\infty \), immediate convergence does not occur, since conflicting high-degree groups require time to reconcile their differing opinions.
The following proposition generalizes Example~\ref{ex-2-e-1-g} to cases with multiple high-degree or multiple low-degree groups. 


\begin{proposition}[Convergence Speed with Three or More Groups]
\label{prop-1vm}
The learning speed, captured by $1/|\lambda_2(\mathbf{T}^*)|$, is generally non-monotonic in the learning heuristic $\alpha$.\footnote{For explicit characterizations of \( |\lambda_2| \) and the threshold(s) of \( \alpha \), see Lemma~\ref{speed:specific} and Proposition~\ref{impact_of_alpha} 
in Appendix~\ref{sec:proofs_convergence_speed}.} 
In particular:
\begin{itemize}
    \item $m_h > 1,\, m_l = 1$: There exist two thresholds, $\alpha_1 < \alpha_2$, such that $|\lambda_2(\mathbf{T}^*)|$ increases with $\alpha$ when $\alpha < \alpha_1$ or $\alpha > \alpha_2$, and decreases with $\alpha$ when $\alpha \in (\alpha_1, \alpha_2)$. In addition, 
    $\lim\limits_{\alpha\rightarrow -\infty}|\lambda_2(\mathbf{T}^*)| = 0$, $\lim\limits_{\alpha\rightarrow +\infty}|\lambda_2(\mathbf{T}^*)| > 0$.
    
    \item $m_h = 1,\, m_l > 1$: There exist two thresholds, $\alpha_1 < \alpha_2$, such that $|\lambda_2(\mathbf{T}^*)|$ decreases with $\alpha$ when $\alpha < \alpha_1$ or $\alpha > \alpha_2$, and increases with $\alpha$ when $\alpha \in (\alpha_1, \alpha_2)$. In addition,
    $\lim\limits_{\alpha\rightarrow -\infty}|\lambda_2(\mathbf{T}^*)| > 0$, $\lim\limits_{\alpha\rightarrow +\infty}|\lambda_2(\mathbf{T}^*)| = 0$.
\end{itemize}
\end{proposition}


\paragraph{Key Takeaway on Convergence Speed}
Our analysis reveals a non-monotonic impact of degree dependence on convergence speed. Greater popularity-favoring does not necessarily lead to faster consensus. Whether it accelerates or slows convergence depends on the diversity of views within high- and low-degree groups. 
When low-degree groups outnumber high-degree ones and different groups hold distinct views, greater popularity-favoring can accelerate convergence; the reverse occurs when high-degree
groups outnumber low-degree groups.

\subsection{Trade-off Between Wisdom and Convergence Speed }
\label{sec:tradeoff}
Our analysis of the impact of degree-weighted learning on societal wisdom and convergence speed reveals a potential trade-off between the two. For instance, while a society that resists popularity bias may achieve greater wisdom ({see Proposition~\ref{cor-wisdom}}), consensus may be reached more slowly when many low-degree groups exist, each potentially holding different views ({see the second case in  Proposition~\ref{prop-1vm}}). Conversely, a more popularity-favoring learning rule may accelerate belief convergence at the cost of wisdom.

Table~\ref{tab:wisdom_speed}, introduced in the Introduction, summarizes the results on societal wisdom ({wise or unwise, Proposition~\ref{cor-wisdom}}) and convergence speed ({fast or slow, Propositions~\ref{prop-1-e-1-g} and~\ref{prop-1vm}}). We discuss potential policy implications of the wisdom–speed trade-off in the concluding remarks.

\section{Discussion on the Symmetry Assumption}  
\label{sec:discussion}
Before concluding, we discuss the symmetry assumption imposed on the adjacency matrix.  Although both directed and undirected networks are natural in a learning setting, we focus on undirected networks because they are required for using the second-largest eigenvalue in modulus (SLEM) to bound convergence speed and, more importantly, to derive our asymptotic result in this context. 

\paragraph{Why Symmetry is Needed} 
The derivation of the inverse relationship between SLEM and convergence speed requires employing spectral decomposition, which generally necessitates that the learning matrix be symmetric (or similar to a symmetric matrix). To prove the asymptotic result, we apply Weyl’s inequalities, which require the matrices to be Hermitian. Moreover, for real matrices, symmetry is both a necessary and sufficient condition for being Hermitian. These technical requirements are the key reasons why we focus on undirected networks.

\paragraph{What Happens When the Network is Asymmetric} 
For non-symmetric matrices, there is no general eigenvalue inequality comparable to Weyl’s inequalities that can effectively bound the difference between $\matr{T}$ and $\matr{T}^*$. The convergence of the high-dimensional random matrix $\matr{T}^{\infty}$ (which is related to a non-symmetric matrix $\matr{A}$) is complex, as we do not know whether $\matr{T}$ is diagonalizable. 
We conjecture that our results still hold for diagonalizable $\matr{T}$.

\paragraph{Simulations}
Consider a stochastic block model with two groups, \(n_1=500\) and \(n_2 = 300\), and the linking matrix \(\textbf{P}'\) given by  
\begin{equation}\label{eq_asy_prob}
    \matr{P}' =
    \begin{pmatrix}
        0.3 & 0.2 + a \\
        0.2 - a & 0.3
    \end{pmatrix},
\end{equation}
where \( a \in [0, 0.2] \) captures the level of asymmetry. When \( a = 0 \), we recover the benchmark of symmetric matrices. 

Figure \ref{fig_asy_speed} plots convergence speeds as a function of the degree dependence \( \alpha \). For small asymmetry \( a \), the convergence speed curves are close to the symmetric benchmark (\( a = 0 \)). When \( a \) is large (e.g., \( a = 0.2 \)), the behavior differs significantly: \( |\lambda_2| \) approaches 1 for high degree dependence \( \alpha \), indicating a failure of convergence.

Figure \ref{fig_asy_convergence} depicts the convergence of \( \matr{T}^t \) and complements the above observation. In particular, we plot the logarithm of the Frobenius norm, \( 2\log\|\matr{T}^{t}-(\matr{Q}^{-1})_{,1} \matr{Q}_{1,}\|_F \), for \( t=1,\dots,20 \) at \( \alpha = 1 \).%
\footnote{If \( \matr{T} \) is diagonalizable, then \( \matr{T} = \matr{Q}^{-1} \Lambda \matr{Q} \), leading to \( \matr{T}^{\infty} = (\matr{Q}^{-1})_{,1} \matr{Q}_{1,} \), where \( \matr{Q}_{1,} \) is the first row of \( \matr{Q} \) and \( (\matr{Q}^{-1})_{,1} \) is the first column of \( \matr{Q}^{-1} \).} 
With mild asymmetry (say \( a \leq 0.05 \)), 
\( \|\matr{T}^{t}-(\matr{Q}^{-1})_{,1} \matr{Q}_{1,}\|_F < e^{-10/2},  t\in[8,20] \), indicating convergence of \( \matr{T}^t \) to the {limit matrix}. With large asymmetry (\( a = 0.18 \) or \( 0.2 \)), 
\( \|\matr{T}^{t}-(\matr{Q}^{-1})_{,1} \matr{Q}_{1,}\|_F > e^{0} = 1,  t\in[1,20] \), meaning \( \matr{T}^t \) does not appear to converge.

\begin{figure}
     \centering
     \begin{subfigure}[b]{0.45\textwidth}
         \centering
\includegraphics[width=\textwidth]{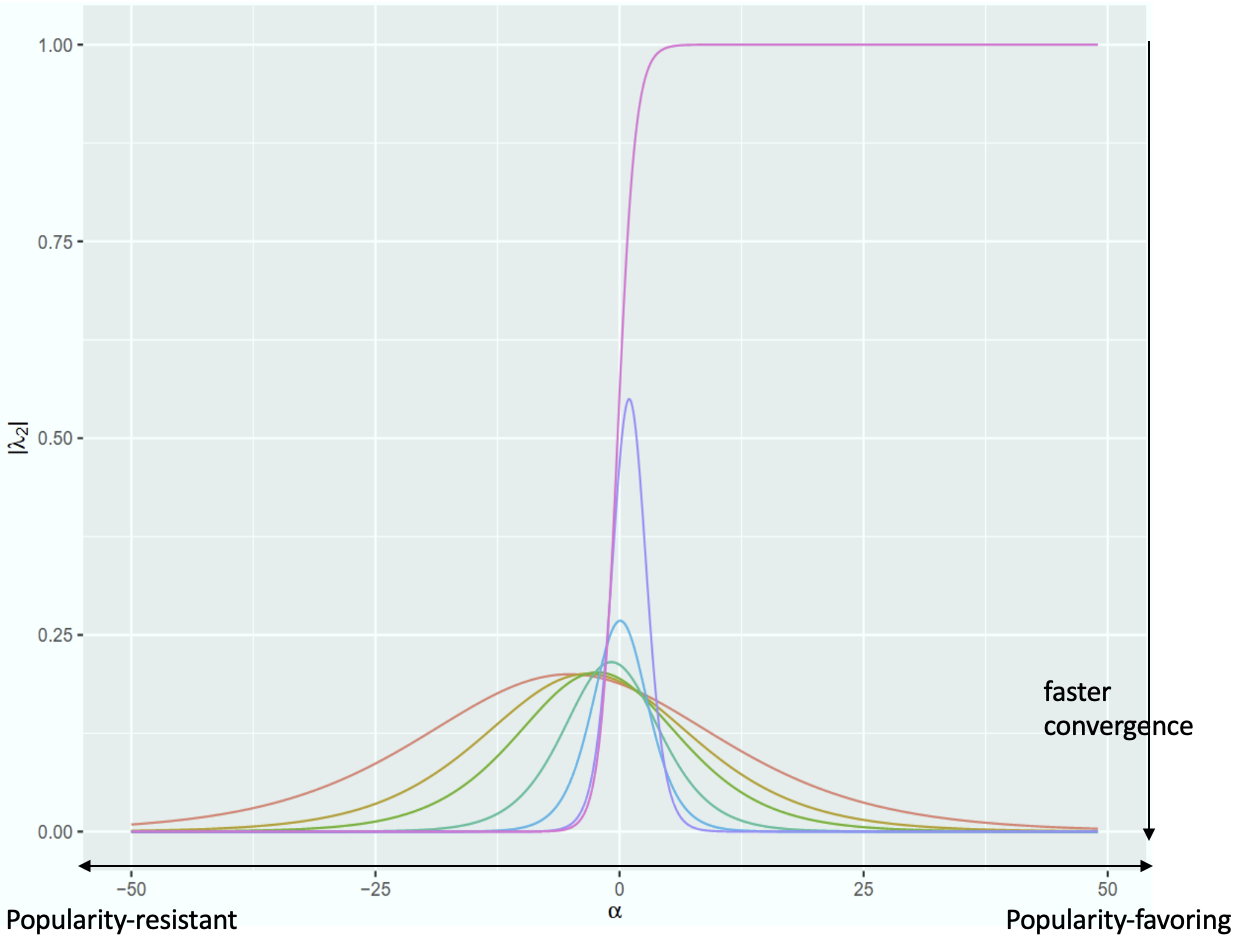}
         \caption{Learning speed, $|\lambda_2|$.}
         \label{fig_asy_speed}
     \end{subfigure}
     \hfill
     \begin{subfigure}[b]{0.45\textwidth}
         \centering
\includegraphics[width=\textwidth]{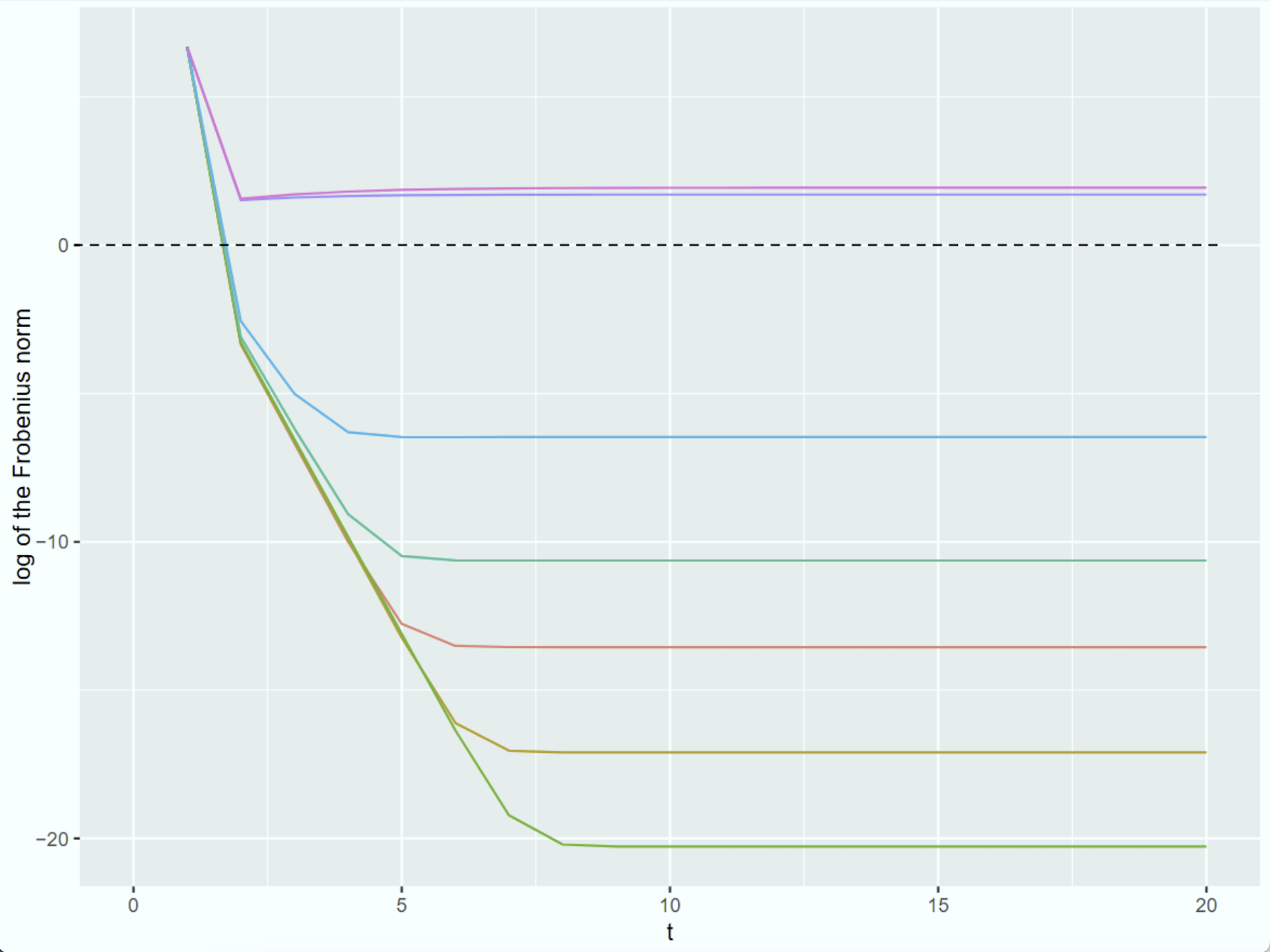}
         \caption{Convergence of $\matr{T}^t$.}
         \label{fig_asy_convergence}
     \end{subfigure}
     \begin{subfigure}[b]{0.08\textwidth}     \includegraphics[width=\textwidth]{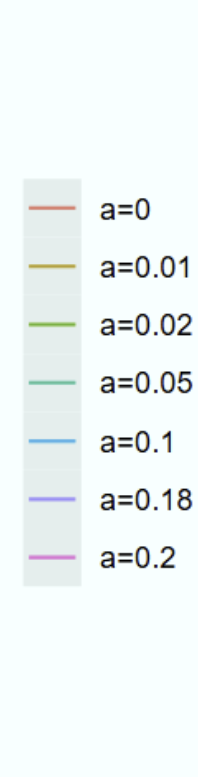}
     \end{subfigure}
     \caption{Learning Speed and Convergence of $\matr{T}^t$ with Varying Asymmetry, $a$.}\label{fig_asy}
     \begin{figurenotes}
There are two groups with group sizes $n_1 = 500$ and $n_2 = 300$; linking probabilities are $p_{11} = p_{22} = 0.3$, $p_{12} = 0.2 + a$, $p_{21} = 0.2 - a$, and $a$ captures the level of asymmetry. Panel (b) plots the logarithm of the Frobenius norm, \(2\log\|\matr{T}^{t}-(\matr{Q}^{-1})_{,1} \matr{Q}_{1,}\|_F\).    \end{figurenotes}
\end{figure}


These simulations convey two key messages. 
\begin{itemize}
    \item A slight departure from symmetry leads to similar learning outcomes, suggesting that the symmetric case is not a knife-edge condition and that our findings exhibit at least local robustness (Figure \ref{fig_asy} with relatively small \( a \)).
    \item A large departure from symmetry, however, introduces substantial differences and challenges. For instance, in symmetric matrices, all eigenvalues are real, and \( |\lambda_2| \) is strictly less than one, whereas in asymmetric matrices, eigenvalues may be complex, and \( |\lambda_2| \) can approach one, indicating extremely slow convergence (Figure \ref{fig_asy} with \( a = 0.18 \) or \( 0.2 \)).
\end{itemize}

\section{Concluding Remarks}
\label{sec:con}
Our views are shaped by neighbors in social networks. How we weigh their opinions affects the duration of disagreement and learning effectiveness in a society.  To address these questions, we propose a degree-weighted DeGroot learning model, in which an agent's weighting of neighbors' opinions is determined by the neighbors' popularity (degree). We examine how degree dependence affects consensus beliefs, societal wisdom, and the speed of belief convergence.

Our main technical contribution is establishing an asymptotic theory of learning outcomes (consensus beliefs and convergence speeds) for degree-weighted learning matrices in random networks. The central economic insight from this framework is that degree dependence shapes both convergence speed and societal wisdom, revealing a trade-off between the two: while a more popularity-favoring rule (higher degree dependence) may accelerate convergence when views among high-degree groups are relatively unified, it may hinder society's ability to learn the true state. 


The economic implications of the trade-off between wisdom and convergence speed offer intriguing avenues for further exploration. For instance, future work could investigate scenarios in which policymakers or platform designers prioritize rapid consensus formation at the expense of societal wisdom. It would also be valuable to characterize optimal threshold values of the degree-dependence in specific informational or social contexts, thereby bridging theoretical predictions with practical applications.


\newpage


\clearpage
\begin{appendix}



\renewcommand{\thesubsection}{\arabic{subsection}}

\numberwithin{proposition}{section}
\numberwithin{lemma}{section}
\numberwithin{theorem}{section}
\numberwithin{equation}{section}


\section{Appendix \Alph{section}. Primitivity of $\matr{R}(n)$, $\matr{T}^*(n)$ and $\matr{A}(n)$}

\label{sec:Tlimit}
\begin{proposition}
\label{remark_connect}
Under Assumptions \ref{ass:density}--\ref{ass:connection}, there exists an integer $n_{0} \in \mathbb{N}$ s.t. $\forall n \ge n_{0}$:
\begin{enumerate}
    \item The expected adjacency matrix $\matr{R}(n)$ and learning matrix $\matr{T}^*(n)$ are primitive.
    \item The realized adjacency matrix $\matr{A}(n)$ is primitive with probability at least $1 - 16n^{-2}$.
\end{enumerate}
\end{proposition}

\paragraph{Proof of Proposition \ref{remark_connect}} First, we establish the primitivity of the expected matrices $\matr{T}^*$ and $\matr{R}$.
Assumption \ref{ass:connection} states that $\limsup_{n\rightarrow\infty}|\lambda_{2}(\matr{T}^{*})|<1$.
By the definition of limit superior, there exists a constant $c < 1$ and an integer $n_1$ such that for all $n \ge n_1$, we have $|\lambda_{2}(\matr{T}^{*}(n))| \le c < 1$.
Since $\matr{T}^*$ is a row-stochastic matrix, its largest eigenvalue is 1. The condition $|\lambda_{2}(\matr{T}^{*})| < 1$ implies that the eigenvalue 1 is simple and strictly dominant in modulus.
For a nonnegative matrix, this condition is sufficient to guarantee primitivity (irreducibility and aperiodicity).
Thus, $\matr{T}^*(n)$ is primitive for all $n \ge n_1$.
Since $\matr{T}^*$ and $\matr{R}$ share the same zero-pattern (connectivity structure) by construction, the primitivity of $\matr{T}^*$ implies that $\matr{R}$ is also primitive for all $n \ge n_1$.

Next, we consider the realized adjacency matrix $\matr{A}$ and learning matrix $\matr{T}$.
By the eigenvalue concentration results established in Part 2 of Theorem \ref{thm:asy},\footnote{We remark that the proof of Part 2 of Theorem \ref{thm:asy} does not rely on the realized matrices $\matr{A}$ and $\matr{T}$ being primitive with a high probability. Therefore, using that result in the present proof does not involve circular reasoning.} 
the second-largest eigenvalue of the realized learning matrix, $|\lambda_{2}(\matr{T})|$, converges in probability to $|\lambda_{2}(\matr{T}^{*})|$.
$\forall n \ge n_1$ (defined in the previous paragraph), since $|\lambda_{2}(\matr{T}^{*})| \leq c < 1$, the spectral gap $1 - |\lambda_{2}(\matr{T}^{*})| \geq 1-c > 0$ is bounded away from zero.
Consequently, there exists an integer $n_{0} \ge n_1$ such that for all $n \ge n_{0}$, the event
\[ E_{n}=\{|\lambda_{2}(\matr{T})|<1\} \]
occurs with probability at least $1-16n^{-2}$.
Conditioning on the event $E_{n}$, since $\matr{T}$ is row-stochastic, the strict inequality $|\lambda_{2}(\matr{T})|<1$ implies that $\matr{T}$ is primitive.
Finally, since each realized $\matr{A}$ and its corresponding $\matr{T}$ share the exact same support (i.e., $\matr{T}_{ij}>0$ if and only if $\matr{A}_{ij}>0$), $\matr{T}$ is primitive if and only if $\matr{A}$ is primitive.
It follows that for all $n \ge n_0$:
\[ \mathbb{P}(\matr{A} \text{ is primitive}) = \mathbb{P}(\matr{T} \text{ is primitive}) \ge \mathbb{P}(E_{n}) \ge 1-16n^{-2}. \]
This completes the proof. \hfill\qed

\section{Appendix \Alph{section}. Proofs of Results in Section \ref{sec:result}}

\subsection{Preliminaries: Generalized Degree-Dependence Function}
\label{subsec:generalized_degree_dependence}


In this section, we prove that our results hold under a more general setting, where  $d^{\alpha}$ is replaced by $ \phi(\alpha, d)$. Here, $ \phi(\alpha, d)$ satisfies Properties~\ref{property_4}-\ref{property_6}.

\begin{property}
\label{property_4}
The function \(\phi(\alpha,d) \in \mathcal{C}^2(\mathbb{R} \times [0,\infty))\) is nonnegative and satisfies \(\phi(0,d) \equiv 1\) for all \( d \in [0,\infty) \). Additionally, \(\phi\) is monotonically increasing in \(\alpha\) for \(\alpha \in \mathbb{R}\), monotonically increasing in \(d\) for \(\alpha \in (0,\infty)\), and monotonically decreasing in \(d\) for \(\alpha \in (-\infty,0)\).
\end{property}

\begin{property}
\label{property_5}
For any \(d_1 > d_2\), the ratio \(\phi(\alpha,d_2)/\phi(\alpha,d_1)\) is strictly decreasing in \(\alpha\):
\begin{equation}
\frac{d}{d\alpha} \left(\frac{\phi(\alpha,d_2)}{\phi(\alpha,d_1)}\right) < 0\,. \label{Property2}
\end{equation}
Furthermore, we have the limits:
\begin{equation}
\lim\limits_{\alpha\rightarrow\infty} \frac{\phi(\alpha,d_2)}{\phi(\alpha,d_1)} = \lim\limits_{\alpha\rightarrow-\infty} \frac{\phi(\alpha,d_1)}{\phi(\alpha,d_2)} = 0.
\end{equation}
\end{property}

\begin{property}
\label{property_6}
The function \(\phi\) satisfies the following conditions:
\begin{equation}
\limsup_{d\rightarrow \infty} \frac{\partial \phi(\alpha,d)/\partial d}{\phi(\alpha,d)/d} < \infty\,,\label{Property3}
\end{equation}
\begin{equation}
\limsup_{d\rightarrow \infty} \frac{\,\partial^2 \phi(\alpha,d)/\partial d^2}{(\partial \phi(\alpha,d)/\partial d)/d} < \infty\,.\label{Property3b}
\end{equation}
\end{property}

Properties \ref{property_4} and \ref{property_5} establish fundamental requirements for \(\phi\) to ensure the validity of our discussion. Property \ref{property_6} consists of two technical conditions necessary for the concentration result derived later for random networks. Intuitively, Property \ref{property_6} ensures that a small change in the degree \(d\) does not lead to an excessive increase in \(\phi\).

\subsection{Proof Preliminaries}

Throughout the appendix, various equalities and inequalities involve constants. We use the capital letter \(C\) to denote such constants, which may vary across expressions but remain positive and independent of \(n\).


\subsubsection{$\matr{T}$ Similar to a Symmetric Matrix} \label{subsec:T_symmetric_similarity}

As a preliminary step, we first establish that \(\matr{T}\) is similar to a symmetric matrix.
\begin{lemma}
\label{lemma:similar}
The learning matrix \(\matr{T}\), as defined in Equation (\ref{NewT}), is similar to a symmetric matrix.
\end{lemma}
\paragraph{Proof of Lemma \ref{lemma:similar}}
The matrix \(\matr{T}\) is given by
\begin{equation}
\matr{T} = \matr{D}_{1,\matr{A}}^{-1} \matr{A} \matr{D}_{2,\matr{A}}\,, \label{MatrixT_appendix}
\end{equation}
where \(\matr{D}_{1,\matr{A}}\) and \(\matr{D}_{2,\matr{A}}\) are diagonal matrices with diagonal entries:
\[
(\matr{D}_{1,\matr{A}})_{ii} = \sum_j A_{ij}\,\phi(\alpha,d_j(\matr{A}))\,,
\ \ \ 
(\matr{D}_{2,\matr{A}})_{ii} = \phi(\alpha,d_i(\matr{A}))\,.
\]
Define the diagonal matrix \(\matr{S}\) as
\begin{equation}
\matr{S} = \left(\matr{D}_{1,\matr{A}} \matr{D}_{2,\matr{A}}\right)^{-1/2}.
\end{equation}
Since the diagonal matrices \(\matr{D}_{1,\matr{A}}\) and \(\matr{D}_{2,\matr{A}}\) commute, we can rewrite \(\matr{T}\) in terms of \(\matr{S}\) as
\begin{equation}
\matr{S}^{-1} \matr{T} \matr{S} = \matr{D}_{2,\matr{A}}^{1/2} \matr{D}_{1,\matr{A}}^{-1/2} \matr{A} \matr{D}_{2,\matr{A}}^{1/2} \matr{D}_{1,\matr{A}}^{-1/2}.
\end{equation}
Since \(\matr{A}\) is symmetric, the right-hand side is also symmetric, proving that \(\matr{T}\) is similar to a symmetric matrix.
\hfill\qed


A direct corollary of the previous result is that \(\matr{T}^*\) can be expressed as  
\begin{equation}
\matr{T}^* = \matr{D}^{-1}_{1,\matr{R}} \matr{R} \matr{D}_{2,\matr{R}}\,,\label{MatrixT*_appendix}
\end{equation}
where \(\matr{R}=\E\matr{A}\), and \(\matr{D}_{1,\matr{R}}\) and \(\matr{D}_{2,\matr{R}}\) are diagonal matrices defined as  
\[
(\matr{D}_{1,\matr{R}})_{ii} = \sum_j R_{ij}\,\phi(\alpha,d_j(\matr{R}))\,,
\ \ \ 
(\matr{D}_{2,\matr{R}})_{ii} = \phi(\alpha,d_i(\matr{R}))\,.
\]

\subsection{Proof of Proposition \ref{prop_expected}}
\label{proof:prop_expected}
\subsubsection{Proof of the First Part of Proposition \ref{prop_expected}}

The following lemma, commonly known as the Perron-Frobenius Theorem, establishes that long-run beliefs converge to a consensus.  
 
\begin{lemma}\label{lemma:consensus}
Let \(\matr{T}\) be a primitive matrix with nonnegative entries. Then the spectral radius \(\rho(\matr{T}) = 1\) is a simple eigenvalue. Moreover, the limit \(\matr{T}^{\infty}\) exists and is given by
{
\[
\matr{T}^{\infty} = \lambda_1 \matr{U}_1 = \vec{v} \vec{w}^{\top}\,,
\]}
where \(\vec{v}\) and \(\vec{w}\) are the left and right unit eigenvectors of \(\matr{T}\) corresponding to \(\lambda_1\).
\end{lemma}

The proof can be found in \cite{Meyer2000}.


\begin{lemma}\label{LemmaA.3}
Let \(\matr{M}_1\) and \(\matr{M}_2\) be similar \(n \times n\) matrices such that  
\[
\matr{M}_2 = \matr{P} \matr{M}_1 \matr{P}^{-1},
\]
where \(\matr{P}\) is an invertible \(n \times n\) matrix. If \(\vec{v}\) is an eigenvector of \(\matr{M}_1\) corresponding to the eigenvalue \(\lambda\), then \(\matr{P}\vec{v}\) is an eigenvector of \(\matr{M}_2\) corresponding to \(\lambda\).
\end{lemma}  
\paragraph{Proof of Lemma \ref{LemmaA.3}}
Since \(\vec{v}\) is an eigenvector of \(\matr{M}_1\) with eigenvalue \(\lambda\), we have  
\(
\matr{M}_1 \vec{v} = \lambda \vec{v}.
\)
Multiplying both sides by \(\matr{P}\) gives  
\(
\matr{P} \matr{M}_1 \vec{v} = \lambda \matr{P} \vec{v}.
\)
Using the similarity relation \(\matr{M}_2 = \matr{P} \matr{M}_1 \matr{P}^{-1}\), we obtain  
\(
\matr{M}_2 (\matr{P} \vec{v}) = \matr{P} \matr{M}_1 \matr{P}^{-1} (\matr{P} \vec{v}) = \matr{P} \matr{M}_1 \vec{v} = \lambda \matr{P} \vec{v}.
\)

Thus, \(\matr{P} \vec{v}\) is an eigenvector of \(\matr{M}_2\) corresponding to \(\lambda\), as desired.
\hfill\qed

We seek to evaluate the limit:
\begin{equation}
\matr{T}^{\infty} = \lim_{t\rightarrow\infty} \matr{T}^{t} = \lim_{t\rightarrow\infty} \matr{D}_{1,\matr{A}}^{-1/2} \matr{D}_{2,\matr{A}}^{-1/2} 
\left(\matr{D}_{2,\matr{A}}^{1/2} \matr{D}_{1,\matr{A}}^{-1/2} \matr{A} \matr{D}_{2,\matr{A}}^{1/2} \matr{D}_{1,\matr{A}}^{-1/2}\right)^{t} 
\matr{D}_{1,\matr{A}}^{1/2} \matr{D}_{2,\matr{A}}^{1/2}. \label{LimitProp2}
\end{equation}
By Lemma \ref{lemma:consensus}, we know that the limit converges:
\begin{equation}
\lim_{t\rightarrow\infty} \left(\matr{D}_{2,\matr{A}}^{1/2} \matr{D}_{1,\matr{A}}^{-1/2} \matr{A} \matr{D}_{2,\matr{A}}^{1/2} \matr{D}_{1,\matr{A}}^{-1/2}\right)^{t} 
= \lambda_1 \vec{v}_1 \vec{v}_1^{\top}, \label{Symmetric_Matrix}
\end{equation}
where \(\lambda_1 = 1\) is the largest eigenvalue in magnitude, and \(\vec{v}_1\) is the corresponding unit eigenvector of the matrix \(\matr{D}_{2,\matr{A}}^{1/2} \matr{D}_{1,\matr{A}}^{-1/2} \matr{A} \matr{D}_{2,\matr{A}}^{1/2} \matr{D}_{1,\matr{A}}^{-1/2}\). 
Since the matrix in (\ref{Symmetric_Matrix}) is symmetric, its left and right eigenvectors are identical. 

To determine \(\vec{v}_1\), we apply Lemma \ref{LemmaA.3}. It is straightforward to verify that \(\matr{D}_{1,\matr{A}}^{-1} \matr{A} \matr{D}_{2,\matr{A}}\) has the eigenvector \(\vec{e} = (1,\dots,1)^{\top}\) corresponding to eigenvalue \(1\). Consequently, the matrix \(\matr{D}_{2,\matr{A}}^{1/2} \matr{D}_{1,\matr{A}}^{-1/2} \matr{A} \matr{D}_{2,\matr{A}}^{1/2} \matr{D}_{1,\matr{A}}^{-1/2}\) has the eigenvector 
\[
\vec{w}_1 = \matr{D}_{2,\matr{A}}^{1/2} \matr{D}_{1,\matr{A}}^{1/2} \vec{e}.
\]
Normalizing \(\vec{w}_1\) to obtain a unit vector, we divide by its Euclidean norm:
\[
\|\vec{w}_1\|_2 = \left\|\matr{D}_{2,\matr{A}}^{1/2} \matr{D}_{1,\matr{A}}^{1/2} \vec{e} \right\|_2 
= \left( \sum_{i,j} A_{ij} \,\phi(\alpha,d_j(\matr{A}))\phi(\alpha,d_i(\matr{A})) \right)^{1/2}.
\]
Thus, the unit eigenvector is given by \(\vec{v}_1 = \vec{w}_1 / \|\vec{w}_1\|_2\), and substituting into (\ref{LimitProp2}) yields

\begin{align*}
\matr{T}^{\infty} &= \matr{D}_{1,\matr{A}}^{-1/2} \matr{D}_{2,\matr{A}}^{-1/2} 
\left( \|\vec{w}_1\|_2^{-1} \matr{D}_{2,\matr{A}}^{1/2} \matr{D}_{1,\matr{A}}^{1/2} \vec{e} \right)
\left( \|\vec{w}_1\|_2^{-1} \matr{D}_{2,\matr{A}}^{1/2} \matr{D}_{1,\matr{A}}^{1/2} \vec{e} \right)^{\top} 
\matr{D}_{1,\matr{A}}^{1/2} \matr{D}_{2,\matr{A}}^{1/2}\\
&= \|\vec{w}_1\|_2^{-2} \matr{E}_n \matr{D}_{1,\matr{A}} \matr{D}_{2,\matr{A}},
\end{align*}
where \(\matr{E}_n\) is an \(n \times n\) matrix with all entries equal to \(1\). This gives the final expression:
\begin{equation}\label{rmk1ref}
\matr{T}^{\infty}_{ij} = \frac{\sum_{i'} A_{i'j} \,\phi(\alpha,d_j(\matr{A}))\phi(\alpha,d_{i'}(\matr{A}))}{\sum_{i',j'} A_{i'j'} \,\phi(\alpha,d_{j'}(\matr{A}))\phi(\alpha,d_{i'}(\matr{A}))}, \quad \forall i,j = 1,\dots,n.
\end{equation}

\hfill\qed

\subsubsection{Proof of the Second Part of Proposition \ref{prop_expected}}
This part generalizes the result in \cite{GolubJackson2012}. For the proof, we establish an upper and a lower bound and show that they coincide. The main structure of the proof follows \cite{GolubJackson2012}. 

A key step in this adaptation is applying the spectral theorem to decompose \(\matr{T}\). Lemma \ref{lemma:similar} establishes that \(\matr{T}\) is similar to a real symmetric matrix, which allows us to apply the spectral theorem for diagonalizable matrices \cite{Meyer2000}. The spectral decomposition of \(\matr{T}\) is given by
\begin{equation}
\matr{T} = \sum_{i=1}^{n} \lambda_i \matr{U}_i\,,\label{Spectral_Decomp}
\end{equation}
where \(\lambda_i\) are the eigenvalues of \(\matr{T}\), ordered in decreasing magnitude, and \(\matr{U}_i\)  is the orthogonal projection associated with \(\lambda_i\). Specifically, we have \(\matr{U}_i=\matr{D}_{2,\matr{A}}^{-1/2} \matr{D}_{1,\matr{A}}^{-1/2}\vec{v}_i\vec{v}_i^T\matr{D}_{2,\matr{A}}^{1/2} \matr{D}_{1,\matr{A}}^{1/2}\), where \(\vec{v}_i\) is the unit eigenvector of the symmetric matrix \(\matr{D}_{2,\matr{A}}^{1/2} \matr{D}_{1,\matr{A}}^{-1/2} \matr{A} \matr{D}_{2,\matr{A}}^{1/2} \matr{D}_{1,\matr{A}}^{-1/2}\) associated with \(\lambda_i \).

\paragraph{Upper Bound} 
We aim to establish an upper bound for the distance \(\left\Vert\left(\matr{T}^t-\matr{T}^{\infty}\right)\vec{b}\right\Vert_{\matr{s}}\) by applying the spectral decomposition of \(\matr{T}\) from (\ref{Spectral_Decomp}). Note that
\begin{equation}
\matr{T}^t-\matr{T}^{\infty}=\sum_{i=2}^n \lambda_i^t \matr{U}_i\,.\label{T^t-T^infty}
\end{equation}
The summation starts from \(i=2\) due to Lemma \ref{lemma:consensus}, which states that the largest eigenvalue of \(\matr{T}\) is \(1\), while all other eigenvalues have magnitudes strictly less than \(1\). The term corresponding to \(i=1\) cancels with \(\matr{T}^{\infty}\) since
\(
\matr{U}_1 = (\vec{v}\vec{w}^{\top})^t = \vec{v}(\vec{w}^{\top}\vec{v})^{t-1}\vec{w}^{\top} = \vec{v}\vec{w}^{\top} = \matr{T}^{\infty}, \quad \text{for } t \in \mathbb{N}.
\)
Applying (\ref{T^t-T^infty}), we obtain
\begin{align}
\left\Vert\left(\matr{T}^t-\matr{T}^{\infty}\right)\vec{b}\right\Vert_{\matr{s}}^2
& = \left\lVert\sum_{i=2}^n \lambda_i^t \matr{U}_i\vec{b}\right\rVert^2_{\matr{s}} \label{spectral}\\
& = \sum_{i=2}^n|\lambda_i|^{2t}\left\lVert\matr{v}_i\matr{v}_i^T\matr{D}_{2,\matr{A}}^{1/2} \matr{D}_{1,\matr{A}}^{1/2} \vec{b}\right\rVert^2_2/D, \quad\quad (\matr{v}_i\matr{v}_j = 0 \text{ for } i\neq j)\nonumber\\
& \leq |\lambda_2|^{2t}\sum_{i=2}^n\left\lVert\matr{v}_i\matr{v}_i^T\matr{D}_{2,\matr{A}}^{1/2} \matr{D}_{1,\matr{A}}^{1/2} \vec{b}\right\rVert^2_2/D\nonumber\\
& = |\lambda_2|^{2t}\left\lVert\sum_{i=2}^n\matr{v}_i\matr{v}_i^T\matr{D}_{2,\matr{A}}^{1/2} \matr{D}_{1,\matr{A}}^{1/2} \vec{b}\right\rVert^2_2/D\,.\quad\quad (\matr{v}_i\matr{v}_j = 0 \text{ for } i\neq j)\nonumber
\end{align}
where $D=\sum\limits_{i',j} A_{i'j}\, d_j^{\alpha}(\matr{A})\, d_{i'}^{\alpha}(\matr{A})$. Noticing that $\sum_{i=2}^n\matr{v}_i\matr{v}_i^T\le \matr{I}$, 
we obtain
\[
\left\Vert\left(\matr{T}^t-\matr{T}^{\infty}\right)\vec{b}\right\Vert_{\matr{s}}^2 \leq |\lambda_2|^{2t} \|\matr{D}_{2,\matr{A}}^{1/2} \matr{D}_{1,\matr{A}}^{1/2}\vec{b}\|_2^2/D = |\lambda_2|^{2t}| \|\vec{b}\|_{\matr{s}}^2=|\lambda_2|^{2t}.
\]
Taking the square root on both sides yields the desired upper bound.
\hfill\qed

\paragraph{Lower Bound} 
To derive a lower bound, we consider a specific choice of \(\vec{b}\)---the eigenvector of \(\matr{T}\) corresponding to \(\lambda_2\) in magnitude, normalized such that \(\|\vec{b}\|_{\matr{s}}^2=1\). Then
\[
\left(\matr{T}^t-\matr{T}^{\infty}\right)\vec{b} = \lambda_2^{\,t} \vec{b}.
\]
This follows from the fact that \(\matr{U}_i \vec{b} = 0\) for all \(i \neq 2\). Consequently,
\begin{align*}
\left\Vert\left(\matr{T}^t-\matr{T}^{\infty}\right)\vec{b}\right\Vert_{\matr{s}}^2
& = \left\lVert\lambda_2^{\,t} \vec{b}\right\rVert^2_{\matr{s}}
 = |\lambda_2|^{2t}.
\end{align*}
Taking the square root on both sides confirms that the lower bound matches the upper bound.

This completes the proof of Proposition \ref{prop_expected}. 
\hfill\qed




\subsection{Proof of Theorem \ref{thm:asy}}
\label{sec:proof_thm_asy}
\subsubsection{Proof of the First Part of Theorem \ref{thm:asy}}\label{b.4.1}
Similar to Equation \eqref{LimitProp2}, we have
$$(\matr{T}^*)^{\infty} = \lim_{t\rightarrow\infty}(\matr{T}^*)^{t}=\lim_{t\rightarrow\infty}\matr{D}_{1,\matr{R}}^{-1/2}\matr{D}_{2,\matr{R}}^{-1/2}\left(\matr{D}_{2,\matr{R}}^{1/2}\matr{R}_{1,\matr{R}}^{-1/2}\matr{R}\matr{D}_{2,\matr{R}}^{1/2}\matr{D}_{1,\matr{R}}^{-1/2}\right)^{t}\matr{D}_{1,\matr{R}}^{1/2}\matr{D}_{2,\matr{R}}^{1/2}\,. $$
Combining this with \eqref{LimitProp2} and \eqref{Symmetric_Matrix}, we have
\begin{align}\label{xha1}
  \|\vec{s}(\matr{T})-\vec{s}(\matr{T}^*)\|\le \|\matr{T}^{\infty}-(\matr{T}^*)^{\infty}\|=& \|\matr{D}_{1,\matr{A}}^{-1/2}\matr{D}_{2,\matr{A}}^{-1/2}\vec{v}_1\vec{v}_1^{\top}\matr{D}_{1,\matr{A}}^{1/2}\matr{D}_{2,\matr{A}}^{1/2}-\matr{D}_{1,\matr{R}}^{-1/2}\matr{D}_{2,\matr{R}}^{-1/2}\vec{v}^*_1\vec{v}_1^{*\top}\matr{D}_{1,\matr{R}}^{1/2}\matr{D}_{2,\matr{R}}^{1/2}\|\nonumber\\
  \le & \|\vec{v}_1\vec{v}_1^{\top}-\vec{v}^*_1\vec{v}_1^{*\top}\|+\|\matr{D}_{1,\matr{A}}^{-1/2}\matr{D}_{2,\matr{A}}^{-1/2}\matr{D}_{1,\matr{R}}^{1/2}\matr{D}_{2,\matr{R}}^{1/2}-\matr{I}\| \nonumber\\
  & +\|\matr{D}_{1,\matr{A}}^{1/2}\matr{D}_{2,\matr{A}}^{1/2}\matr{D}_{1,\matr{R}}^{-1/2}\matr{D}_{2,\matr{R}}^{-1/2}-\matr{I}\|,
\end{align}
where $\vec{v}^*_1$ is the first eigenvector of $\matr{D}_{2,\matr{R}}^{1/2}\matr{D}_{1,\matr{R}}^{-1/2}\matr{R}\matr{D}_{2,\matr{R}}^{1/2}\matr{D}_{1,\matr{R}}^{-1/2}$.

In order to control the right-hand side of the inequality, we introduce the following propositions and lemmas, with their proofs presented in Online Appendix~O1

\begin{proposition}
\label{propC.1}
$\|\matr{A}-\matr{R}\|\leq 3\sqrt{n\log n}\,,$
with probability at least $1-{2}/{n^2}$.
\end{proposition}

\begin{proposition}
\label{propC.2}
Under Assumptions \ref{ass:density}-\ref{ass:cdensities}, with probability at least \(1-{14}/{n^2}\),
\begin{equation}
\left\Vert\matr{D}_{\matr{A}}^{-1/2}-\matr{D}_{\matr{R}}^{-1/2}\right\Vert\leq C\sqrt{n\log n}\left(\min_i d_i(\matr{R})\right)^{-3/2}\,.
\end{equation}
\end{proposition}

\begin{lemma}\label{LemmaC.1}
Let \(K>1\) be a constant independent of \(n\). Then, for \(\alpha\geq 0\),
\begin{equation}
\phi(\alpha,\min_i d_i(\matr{R}))\leq \phi(\alpha,K\min_i d_i(\matr{R})) \leq C_1\phi(\alpha,\min_i d_i(\matr{R}))\,,\label{IneqC.1.1}
\end{equation}
\begin{equation}
C_2\frac{\partial \phi}{\partial d}(\alpha,\min_i d_i(\matr{R})) \leq \frac{\partial \phi}{\partial d}(\alpha,K\min_i d_i(\matr{R}))\leq C_3\frac{\partial \phi}{\partial d}(\alpha,\min_i d_i(\matr{R}))\,,\label{IneqC.1.2}
\end{equation}
where \(C_1, C_2, C_3\) are positive constants independent of \(n\). 

On the other hand, for \(\alpha<0\),
\begin{equation}
{C_4}\phi(\alpha,\min_i d_i(\matr{R}))\leq \phi(\alpha,K\min_i d_i(\matr{R})) \leq \phi(\alpha,\min_i d_i(\matr{R}))\,,\label{IneqC.1.3}
\end{equation}
\begin{equation}
{C_5}\frac{\partial \phi}{\partial d}(\alpha,\min_i d_i(\matr{R})) \leq \frac{\partial \phi}{\partial d}(\alpha,K\min_i d_i(\matr{R}))\leq {C_6}\frac{\partial \phi}{\partial d}(\alpha,\min_i d_i(\matr{R}))\,,\label{IneqC.1.4}
\end{equation}
{where \(C_4, C_5, C_6\) are positive constants independent of \(n\). 
}

\end{lemma}

\begin{lemma}\label{LemmaC.2}
For a fixed \(j\), with probability at least \(1-{2}/{n^4}\),
\begin{equation}
|d_j(\matr{A})-d_j(\matr{R})|\leq \sqrt{2n\log n}\,.\label{IneqC.2.1}
\end{equation}
Moreover, for some \(c_{j,1}\in \left[d_j(\matr{R})-\sqrt{2n\log n},d_j(\matr{R})+\sqrt{2n\log n}\right]\), \(j = 1,...,n\), with probability at least \(1-{2}/{n^4}\),
\begin{equation}
\left\vert \phi(\alpha,d_j(\matr{A}))-\phi(\alpha,d_j(\matr{R}))\right\vert \leq \sqrt{2n\log n}\,\left|\frac{\partial\phi}{\partial d}(\alpha,c_{j,1})\right|\,.\label{IneqC.2.2}
\end{equation}
\end{lemma}

We first control the second term on the right-hand side of \eqref{xha1}. By Lemmas \ref{LemmaC.1}--\ref{LemmaC.2} and Proposition \ref{propC.2}, with probability at least \(1-{15}/{n^2}\), we obtain, for some positive constant \(C\),
\begin{align}\label{xha2}
\|\matr{D}_{1,\matr{A}}^{-1/2}\matr{D}_{2,\matr{A}}^{-1/2}\matr{D}_{1,\matr{R}}^{1/2}\matr{D}_{2,\matr{R}}^{1/2}-\matr{I}\|
&\leq \|\matr{D}_{\matr{A}}^{-1/2}\matr{D}_{\matr{R}}^{1/2}-\matr{I}\|\|\matr{D}_{2,\matr{R}}\matr{D}_{2,\matr{A}}^{-1}\|+\|\matr{D}_{2,\matr{R}}\matr{D}_{2,\matr{A}}^{-1}-\matr{I}\| \\
&\leq C\sqrt{n\log n}(\min_i d_i(\matr{R}))^{-1},\nonumber
\end{align}
A similar inequality holds for the third term on the right-hand side of \eqref{xha1}.

For the first term on the right-hand side of \eqref{xha1}, applying the Davis-Kahan theorem \citep{davis,tengyao} and using Assumption~\ref{ass:connection} (which ensures that the spectral gap \(1 - \lambda_2(\matr{T}^*)\) is eventually bounded below by some positive constant $1-c>0$), we obtain
\begin{equation}
\|\vec{v}_1\vec{v}_1^{\top}-\vec{v}^*_1\vec{v}_1^{*\top}\|
\leq C\|\matr{D}_{2,\matr{A}}^{1/2}\matr{D}_{1,\matr{A}}^{-1/2}\matr{A}\matr{D}_{2,\matr{A}}^{1/2}\matr{D}_{1,\matr{A}}^{-1/2}
-\matr{D}_{2,\matr{R}}^{1/2}\matr{D}_{1,\matr{R}}^{-1/2}\matr{R}\matr{D}_{2,\matr{R}}^{1/2}\matr{D}_{1,\matr{R}}^{-1/2}\|,
\end{equation}
where $C$ is a positive constant inversely related to the spectral gap $1-c > 0$. Furthermore, using a similar argument as in the proof of \eqref{xha2}, and applying Propositions \ref{propC.1}--\ref{propC.2} and Lemmas \ref{LemmaC.1}--\ref{LemmaC.2} to control the differences $\|\matr{D}_{2,\matr{A}}^{1/2}-\matr{D}_{2,\matr{R}}^{1/2}\|$, $\|\matr{D}_{1,\matr{A}}^{-1/2}-\matr{D}_{1,\matr{R}}^{-1/2}\|$, and $\|\matr{A}-\matr{R}\|$, we find that with probability at least \(1-{15}/{n^2}\),
\begin{equation}
\|\matr{D}_{2,\matr{A}}^{1/2}\matr{D}_{1,\matr{A}}^{-1/2}\matr{A}\matr{D}_{2,\matr{A}}^{1/2}\matr{D}_{1,\matr{A}}^{-1/2}
-\matr{D}_{2,\matr{R}}^{1/2}\matr{D}_{1,\matr{R}}^{-1/2}\matr{R}\matr{D}_{2,\matr{R}}^{1/2}\matr{D}_{1,\matr{R}}^{-1/2}\|
\leq \tilde{C}\sqrt{n\log n}(\min_i d_i(\matr{R}))^{-1} = \tilde{C}\frac{\sqrt{\log n}}{\tau_n\sqrt{n}}
\end{equation}
for some $\tilde{C}>0$. This completes the proof of the first part of Theorem \ref{thm:asy}. \hfill\qed

\subsubsection{Proof of the Second  Part of Theorem \ref{thm:asy}}
\label{sec:proofthm2nd}
Recall that \(\matr{R}=\E\matr{A}\). From the definitions of the matrices in Equations (\ref{MatrixT_appendix}) and (\ref{MatrixT*_appendix}), and by Lemma \ref{LemmaA.3}, comparing the eigenvalues \(\lambda_2(\matr{T})\) and \(\lambda_2(\matr{T}^*)\) reduces to comparing the eigenvalues of 
\(\matr{D}_{2,\matr{A}}^{1/2}\matr{D}_{1,\matr{A}}^{-1/2}\matr{A}\matr{D}_{2,\matr{A}}^{1/2}\matr{D}_{1,\matr{A}}^{-1/2}
\)
and
\(
\matr{D}_{2,\matr{R}}^{1/2}\matr{D}_{1,\matr{R}}^{-1/2}\matr{R}\matr{D}_{2,\matr{R}}^{1/2}\matr{D}_{1,\matr{R}}^{-1/2},
\)
since these matrices are similar to \(\matr{T}\) and \(\matr{T}^*\), respectively.
For simplicity, define
$$
\matr{D}_{\matr{A}} = \matr{D}_{2,\matr{A}}^{-1}\matr{D}_{1,\matr{A}}
\quad \text{and} \quad
\matr{D}_{\matr{R}} = \matr{D}_{2,\matr{R}}^{-1}\matr{D}_{1,\matr{R}}.
$$
By Weyl's inequality, we obtain
\begin{equation}
|\lambda_2(\matr{T})-\lambda_2(\matr{T}^*)| \leq \left\lVert\matr{D}_{\matr{A}}^{-1/2}\matr{A}\matr{D}_{\matr{A}}^{-1/2}-\matr{D}_{\matr{R}}^{-1/2}\matr{R}\matr{D}_{\matr{R}}^{-1/2}\right\rVert.\label{Weyl's}
\end{equation}
Therefore, it suffices to bound the spectral norm, \(\|\cdot\|\), on the right-hand side.


Furthermore, under Assumptions \ref{ass:density} and \ref{ass:cdensities} we have 
$$
0 < \liminf_n \frac{\min_i d_i(\matr{R})}{\max_i d_i(\matr{R})} \leq \limsup_n \frac{\max_i d_i(\matr{R})}{\min_i d_i(\matr{R})} <\infty,
$$
to ensure that $d_j(\mathbf{R}) / d_i(\mathbf{R})$ is bounded by positive constants $K_1, K_2$ for any $i, j$. Combining with the properties of $\phi$ (as utilized in Lemma \ref{LemmaC.1}, Appendix \ref{b.4.1}), it follows that the ratio $\phi\left(\alpha, d_j(\mathbf{R})\right) / \phi\left(\alpha, d_i(\mathbf{R})\right)$ is bounded by positive constants $C_{\phi, \min }$ and $C_{\phi, \max }$ (which depend on $\alpha, K_1, K_2$, but not on $n$ or $d_i(\mathbf{R})$ itself). Thus, for each $i$:
$$
\left(\mathbf{D}_{\mathbf{R}}\right)_{i i}=\sum_j R_{i j} \frac{\phi\left(\alpha, d_j(\mathbf{R})\right)}{\phi\left(\alpha, d_i(\mathbf{R})\right)} \geq\left(\sum_j R_{i j}\right) C_{\phi, \min }=d_i(\mathbf{R}) \cdot C_{\phi, \min }.
$$
Therefore, we have, for some positive constant $C_{\phi, \min }$ that is  independent of $n$,
$$
\min _i\left(\mathbf{D}_{\mathbf{R}}\right)_{i i} \geq C_{\phi, \min } \cdot \min _i d_i(\mathbf{R}).
$$
Then, note that the spectral norm in (\ref{Weyl's}) can be decomposed as:
\begin{align}\label{0705.x2}
&\left\lVert\matr{D}_{\matr{A}}^{-1/2}\matr{A}\matr{D}_{\matr{A}}^{-1/2}-\matr{D}_{\matr{R}}^{-1/2}\matr{R}\matr{D}_{\matr{R}}^{-1/2}\right\rVert \\ \nonumber
\leq & \left\lVert\matr{D}_{\matr{A}}^{-1/2}\matr{A}\matr{D}_{\matr{A}}^{-1/2}-\matr{D}_{\matr{R}}^{-1/2}\matr{A}\matr{D}_{\matr{A}}^{-1/2}\right\rVert
+\left\lVert\matr{D}_{\matr{R}}^{-1/2}\matr{A}\matr{D}_{\matr{A}}^{-1/2}-\matr{D}_{\matr{R}}^{-1/2}\matr{A}\matr{D}_{\matr{R}}^{-1/2}\right\rVert \\ \nonumber
&+\left\lVert\matr{D}_{\matr{R}}^{-1/2}\matr{A}\matr{D}_{\matr{R}}^{-1/2}-\matr{D}_{\matr{R}}^{-1/2}\matr{R}\matr{D}_{\matr{R}}^{-1/2}\right\rVert \\ \nonumber
\leq & \left\Vert\matr{D}_{\matr{A}}^{-1/2}-\matr{D}_{\matr{R}}^{-1/2}\right\Vert \cdot \|\matr{A}\|\cdot \left\Vert\matr{D}_{\matr{A}}^{-1/2}\right\Vert
+\left\Vert\matr{D}_{\matr{A}}^{-1/2}-\matr{D}_{\matr{R}}^{-1/2}\right\Vert \cdot \|\matr{A}\|\cdot \left\Vert\matr{D}_{\matr{R}}^{-1/2}\right\Vert \\ \nonumber
&+\|\matr{A}-\matr{R}\|\cdot \left\Vert\matr{D}_{\matr{R}}^{-1}\right\Vert.\nonumber
\end{align}
By the triangle inequality,
\begin{equation}
\|\matr{A}\|\leq \|\matr{A}-\matr{R}\|+\|\matr{R}\|\,.
\end{equation}
The spectral norm is bounded above by the Frobenius norm:
\begin{equation}
\|\matr{R}\| \leq \|\matr{R}\|_F = \sqrt{\sum_{i,j}R^2_{ij}}\leq Cn \max_{i,j}p_{ij} \leq C\min_j d_j(\matr{R})\,,
\end{equation}
for all \(n>0\), where \(\max_{i,j}p_{ij}\) is the largest entry of the matrix \(\matr{P}\) in the stochastic block model.

Applying Proposition \ref{propC.1} along with Assumption \ref{ass:density}, we obtain
\begin{equation}
\|\matr{A}\|\leq C\min_j d_j(\matr{R})\,,
\end{equation}
with probability at least \(1-{2}/{n^2}\). Similarly, by Proposition \ref{propC.2},
\begin{equation}
\left\lVert\matr{D}_{\matr{A}}^{-1/2}\right\rVert \leq C \left(\min_j d_j(\matr{R})\right)^{-1/2}\,,
\end{equation}
with probability at least \(1-{14}/{n^2}\), for sufficiently large \(n\).

Conditioning on the event
\begin{equation}
\left\lbrace\|\matr{A}-\matr{R}\|\leq 3\sqrt{n\log n}\right\rbrace
\bigcap
\left\lbrace\left\Vert\matr{D}_{\matr{A}}^{-1/2}-\matr{D}_{\matr{R}}^{-1/2}\right\Vert\leq C\sqrt{n\log n}\left(\min_i d_i(\matr{R})\right)^{-3/2}\right\rbrace,
\end{equation}
which holds with probability at least \(1-{16}/{n^2}\), we derive
\begin{equation}
\left\lVert\matr{D}_{\matr{A}}^{-1/2}\matr{A}\matr{D}_{\matr{A}}^{-1/2}-\matr{D}_{\matr{R}}^{-1/2}\matr{R}\matr{D}_{\matr{R}}^{-1/2}\right\rVert
\leq C\frac{\sqrt{n\log n}}{\min_i d_i(\matr{R})} = C\frac{\sqrt{\log n}}{\tau_n\sqrt{n}}\,.\label{tailboundTheorem1}
\end{equation}
This completes the proof for symmetric \(\matr{A}\) and concludes our proof of Theorem \ref{thm:asy}.
\hfill\qed

\paragraph{Proof of Remark \ref{rm:alpha}}
\label{proof:rm_alpha}
The proof is similar to Proposition~\ref{prop_expected} and Theorem~\ref{thm:asy}. The main difference lies in controlling \( |d_{i}^{\alpha_{i}} - d_{i}^{\alpha_{i}}(\matr{R})| \) for varying \( \alpha_i \). Since \( \max_{i,j}|\alpha_i - \alpha_j| \le c \), the values of \( \alpha_i \) concentrate around \( \bar{\alpha} = \sum_{i=1}^n \alpha_i/n \), implying that \( |d_{i}^{\alpha_{i}} - d_{i}^{\alpha_{i}}(\matr{R})| \) is of the same order as \( |d_{i}^{\bar{\alpha}} - d_{i}^{\bar{\alpha}}(\matr{R})| \).
\hfill\qed

\section{Appendix \Alph{section}. Proofs of Results in Section \ref{sec:app}}

\subsection{Wisdom: Related Proofs}
\label{sec:proofs_wisdom}

\paragraph{Proof of Proposition \ref{prop-influence}} 
\label{proof:prop_influence}
Given any two-type model and the corresponding expected network $\matr{R}$, let $d_{\theta \theta'}$ be the expected number of type-$\theta'$ agents linked to each type-$\theta$ agent, $\theta, \theta' \in \{h, l\}$. 

Recall that there are $n_\theta$ type-$\theta$ agents. By Equation~(\ref{eq_influence}) in Proposition~\ref{prop_expected}, a type-$\theta$ agent's influence is 
\[
s_\theta = \frac{S_\theta}{n_h S_h + n_l S_l},
\]
where%
\vspace{-5mm}
\[
S_h = d_{hh} d_h^\alpha d_h^\alpha + d_{hl} d_h^\alpha d_l^\alpha, 
\quad \text{and} \quad
S_l = d_{ll} d_l^\alpha d_l^\alpha + d_{lh} d_h^\alpha d_l^\alpha.
\]

It is easy to see that
\[
\frac{S_h}{S_l}
= \frac{d_{hh} d_h^\alpha d_h^\alpha + d_{hl} d_h^\alpha d_l^\alpha}{d_{ll} d_l^\alpha d_l^\alpha + d_{lh} d_h^\alpha d_l^\alpha}
= \frac{d_{hh} (d_h/d_l)^\alpha + d_{hl}}{d_{ll} (d_h/d_l)^{-\alpha} + d_{lh}}.
\]

Since $d_h > d_l$ (by definition)
and $d_{hh}, d_{hl}, d_{lh}, d_{ll} > 0$ are constant in $\alpha$, we have
${S_h}/{S_l}$ increases in $\alpha$, 
diverges to $+\infty$ as $\alpha \rightarrow +\infty$, and
converges to $0$ as $\alpha \rightarrow -\infty$.

Thus, 
\[
s_h = \frac{S_h}{n_h S_h + n_l S_l} 
= \frac{1}{n_h + n_l (S_h/S_l)^{-1}}
\]
increases in $\alpha$, and
\[
s_l = \frac{S_l}{n_h S_h + n_l S_l} 
= \frac{1}{n_h (S_h/S_l) + n_l}
\]
decreases in $\alpha$.
Finally, $s_h(-\infty) = 0$, $s_l(-\infty) = 1/n_l$, $s_h(+\infty) = 1/n_h$, and $s_l(+\infty) = 0$. \hfill\qed

\paragraph{Proof of Proposition \ref{cor-wisdom}}
\label{proof:cor_wisdom}
The result follows immediately from Corollary~\ref{cor:influence} and Proposition~\ref{prop-influence}.
\hfill\qed

\subsection{Convergence Speed: Explicit Characterization and Proofs}
\label{sec:proofs_convergence_speed}

Recall that in Section \ref{sec:speed} we assume that there are $m$ groups such that $n_1 \neq n_2 = \ldots = n_m$, with linking probabilities $p_{kk} = p$ and $p_{kl} = q$, where $p \neq q$. The expected degrees are then given by \( d_1^* = n_1p+(m-1)n_2q \) and \( d_2^* = n_1q+n_2p+(m-2)n_2q \).
For any weighting function \( \phi \) that satisfies Properties $\ref{property_4}$--$\ref{property_6}$, let $g(\alpha)$ denote the relative weight, as follows:
\[
g(\alpha) = \frac{\phi(\alpha,d_2^*)}{\phi(\alpha,d_1^*)}.
\]

We obtain the following explicit characterization of convergence speed,  
which immediately implies Propositions \ref{prop-1-e-1-g} and \ref{prop-1vm} in the main text (Section \ref{sec:speed}).

\begin{lemma} 
\label{speed:specific} The speed of convergence, negatively related to \( |\lambda_2| \), is as follows.
\begin{itemize}
    \item When \( m=2 \) ($m_h = m_l = 1$),
    \begin{equation}
    \label{lambda2m2_g}
    |\lambda_2(\matr{T}^*)|=\left|
    \frac{n_1p\,\phi(\alpha,d_1^*)}{n_1p\,\phi(\alpha,d_1^*)+n_2q\,\phi(\alpha,d_2^*)}
    -\frac{n_1q\,\phi(\alpha,d_1^*)}{n_1q\,\phi(\alpha,d_1^*)+n_2p\phi(\alpha,d_2^*)}
    \right|.
    \end{equation}
    
    \item When \( m\ge 3 \) and \( d_1^* < d_2^* \) ($m_h > 1,  m_l = 1$),
    \begin{equation}
    |\lambda_2(\matr{T}^*)|=
    \begin{cases}
    \left|\frac{n_1p\,\phi(\alpha,d_1^*)}{n_1p\,\phi(\alpha,d_1^*)+(m-1)n_2q\,\phi(\alpha,d_2^*)}
    -\frac{n_1q\,\phi(\alpha,d_1^*)}{n_1q\,\phi(\alpha,d_1^*)+(n_2p+(m-2)n_2q)\phi(\alpha,d_2^*)}\right|,
    &\alpha\leq g^{-1}\!\left(\tfrac{n_1}{n_2}\right),\\[1ex]
    \left|\frac{n_2(p-q)\,\phi(\alpha,d_2^*)}{n_1q\,\phi(\alpha,d_1^*)+(n_2p+(m-2)n_2q)\phi(\alpha,d_2^*)}\right|,
    &\alpha> g^{-1}\!\left(\tfrac{n_1}{n_2}\right),
    \end{cases}
    \label{lambda2m3_2}
    \end{equation}
    
    \item When \( m\ge 3 \) and \( d_1^* > d_2^* \) ($m_h = 1, m_l > 1$), 
    \begin{equation}
    |\lambda_2(\matr{T}^*)|=
    \begin{cases}
    \left|\frac{n_1p\,\phi(\alpha,d_1^*)}{n_1p\,\phi(\alpha,d_1^*)+(m-1)n_2q\,\phi(\alpha,d_2^*)}
    -\frac{n_1q\,\phi(\alpha,d_1^*)}{n_1q\,\phi(\alpha,d_1^*)+(n_2p+(m-2)n_2q)\phi(\alpha,d_2^*)}\right|,
    &\alpha\geq g^{-1}\!\left(\tfrac{n_1}{n_2}\right),\\[1ex]
    \left|\frac{n_2(p-q)\,\phi(\alpha,d_2^*)}{n_1q\,\phi(\alpha,d_1^*)+(n_2p+(m-2)n_2q)\phi(\alpha,d_2^*)}\right|,
    &\alpha< g^{-1}\!\left(\tfrac{n_1}{n_2}\right).
    \end{cases}
    \label{lambda2m3_1}
    \end{equation}
\end{itemize}
\end{lemma}

\paragraph{Proof of Lemma \ref{speed:specific}} The proof involves lengthy calculations and is provided in Online Appendix~O2.
\hfill\qed\\

With \( |\lambda_2| \) fully characterized, we now analyze how  \( \alpha \) affects convergence speed.

\begin{proposition}[Non-monotonic Impact of Degree-Dependence, \(\alpha\)]
\label{impact_of_alpha}
Let $\mathcal{D}_\alpha$ be the region where \(|\lambda_2(\mathbf{T}^*)|\) is decreasing in \(\alpha\) for \(\alpha \in \mathcal{D}_\alpha\) and increasing for \(\alpha \notin \mathcal{D}_\alpha\). Then, \(\mathcal{D}_\alpha\) is as follows:
\begin{equation}\label{dalpha}
\mathcal{D}_\alpha=\begin{cases}
\big[g^{-1}\!\left(\tfrac{n_1}{n_2}\right),\infty\big), & \text{if } m=2, \\[1ex]
(-\infty,g^{-1}\!\left(\tfrac{n_1}{n_2}\right)]\cup\bigg[g^{-1}\!\left(\tfrac{n_1}{n_2}\Big(\tfrac{p}{(m-1)(p+(m-2)q)}\Big)^{1/2}\right),\infty\bigg), & \text{if } m\geq 3, \, d_1^*>d_2^*, \\[1ex]
\bigg(g^{-1}\!\left(\tfrac{n_1}{n_2}\Big(\tfrac{p}{(m-1)(p+(m-2)q)}\Big)^{1/2}\right),g^{-1}\!\left(\tfrac{n_1}{n_2}\right)\bigg], & \text{if } m\geq 3, \, d_1^*<d_2^*.
\end{cases}
\end{equation}
\end{proposition}

\paragraph{Proof of Proposition \ref{impact_of_alpha}} See Online Appendix~O2 
for details.
\hfill\qed

\end{appendix}

\newpage

\bibliographystyle{aea}

\bibliography{Social_Learning_2023.bib}


\end{document}